\documentclass[12pt, draftclsnofoot, onecolumn]{IEEEtran}
\usepackage{amsthm}
\usepackage{subfig}
\usepackage{amsmath}
\usepackage{comment}
\usepackage{amsfonts}
\usepackage{amssymb}
\ifCLASSINFOpdf
  \usepackage[pdftex]{graphicx}
  \usepackage{graphicx}
  \graphicspath{{../pdf/}{../jpeg/}{../png/}}
  \DeclareGraphicsExtensions{.pdf,.jpeg,.png}
\else
  \usepackage[dvips]{graphicx}
  \graphicspath{{../eps/}}
  \DeclareGraphicsExtensions{.eps}
\fi
\usepackage{multirow}
\usepackage{epstopdf}
\usepackage{float}
\usepackage{xcolor}

\newtheorem{lem}{Lemma}
\usepackage{cite}
\ifCLASSINFOpdf
   \usepackage[pdftex]{graphicx}
   \else
\usepackage[dvips]{graphicx}
\fi
\usepackage[T1]{fontenc}
\newcommand*{\affmark}[1][*]{\textsuperscript{\dag}}
\usepackage{algorithm}
\usepackage[noend]{algpseudocode}
\makeatletter
\def\BState{\State\hskip-\ALG@thistlm}
\makeatother
\usepackage{acronym}
\acrodef{THz}[THz]{Tera-Hertz}
\acrodef{MIMO}[MIMO]{multiple-input multiple-output}
\acrodef{AWGN}[AWGN]{additive white Gaussian noise}
\acrodef{LoS}[LoS]{line-of-sight}
\acrodef{NLoS}[NLoS]{non-line-of-sight}
\acrodef{BS}[BS]{base station}
\acrodef{UE}[UE]{user equipment}
\acrodef{AoD}[AoD]{angle of departure}
\acrodef{AoA}[AoA]{angle of arrival}
\acrodef{ULA}[ULA]{uniform linear array}
\acrodef{HITRAN}[HITRAN]{HIgh resolution TRANsmission}
\acrodef{i.i.d.}[i.i.d.]{independent and identically distributed}
\acrodef{RF}[RF]{radio frequency}
\acrodef{LS}[LS]{Least Square}
\acrodef{MMSE}[MMSE]{Minimum Mean Square Error}
\acrodef{OMP}[OMP]{
orthogonal matching pursuit}
\acrodef{BL}[BL]{Bayesian learning}
\acrodef{EM}[EM]{Expectation-Maximization}
\acrodef{E-step}[E-step]{Expectation-step}
\acrodef{M-step}[M-step]{Maximization-step}
\acrodef{BCRB}[BCRB]{Bayesian Cramer-Rao Bound}
\acrodef{BCRLB}[BCRLB]{Bayesian Cramer-Rao Lower Bound}
\acrodef{MSE}[MSE]{mean square error}
\acrodef{FIM}[FIM]{Fisher information matrix}
\acrodef{pdf}[PDF]{Probability Density Function}
\acrodef{CSI}[CSI]{channel state information}
\acrodef{SVD}[SVD]{singular value decomposition}
\acrodef{DFT}[DFT]{discrete Fourier transform}
\acrodef{BER}[BER]{bit-error rate}
\acrodef{Tbps}[Tbps]{Tera-bits per second }
\acrodef{AR}[AR]{augmented reality}
\acrodef{VR}[VR]{virtual reality}
\acrodef{SNR}[SNR]{signal-to-noise power ratio} 
\acrodef{mmWave}[mmWave]{milli-meter wave}
\acrodef{APSs}[APSs]{analog phase shifters}
\acrodef{DSP}[DSP]{digital signal processor}
\acrodef{PA}[PA]{\textit{a priori} aided }
\acrodef{DLA}[DLA]{discrete lens antenna}
\acrodef{MMV}[MMV]{multiple measurement vectors}
\acrodef{MBL}[MBL]{MMV-BL}
\acrodef{AMP}[AMP]{approximate message passing }
\acrodef{SISO}[SISO]{single-input single-output}
\acrodef{IRS}[IRS]{intelligent reflecting surface}
\acrodef{QPSK}[QPSK]{quadrature-phase shift keying}
\acrodef{NMSE}[NMSE]{normalized MSE}
\acrodef{FOCUSS}[FOCUSS]{FOCal Underdetermined System Solver}
\acrodef{ASE}[ASE]{achievable spectral-efficiency}
\acrodef{OFDM}[OFDM]{orthogonal frequency division multiplexing}

\begin{document}
\title{\Large Hybrid Transceiver Design for Tera-Hertz MIMO Systems Relying on Bayesian Learning Aided Sparse Channel Estimation}\vspace{-5pt}
\author{\small Suraj~Srivastava, \textit{Student Member,~IEEE,} Ajeet~Tripathi, Neeraj~Varshney, \textit{Member,~IEEE,} \\ Aditya~K.~Jagannatham, \textit{Member,~IEEE,} and Lajos~Hanzo,  \textit{Fellow,~IEEE}
\thanks{S. Srivastava, A Tripathi and A. K. Jagannatham are with the Department of Electrical Engineering, Indian Institute of Technology Kanpur, India-208016 (e-mail: \{ssrivast, tajeet, adityaj\}@iitk.ac.in). N. Varshney is with the Wireless Networks Division, National Institute of Standards and Technology (NIST), Gaithersburg, MD 20899-6730, USA (E-mail: neerajv@ieee.org). L. Hanzo is with the School of Electronics and Computer Science, University of Southampton, Southampton SO17 1BJ, U.K. (e-mail: lh@ecs.soton.ac.uk).
 }\vspace{-25pt}
}
\maketitle
\vspace{-25pt}
\begin{abstract}
Hybrid transceiver design in \ac{MIMO} \ac{THz} systems relying on sparse \ac{CSI} estimation techniques is conceived. To begin with, a practical \ac{MIMO} channel model is developed for the \ac{THz} band that incorporates its molecular absorption and reflection losses, as well as its \ac{NLoS} rays associated with its diffused components. Subsequently, a novel \ac{CSI} estimation model is derived by exploiting the angular-sparsity of the \ac{THz} \ac{MIMO} channel. Then an \ac{OMP}-based framework is  conceived, followed by designing a sophisticated \ac{BL}-based approach for efficient estimation of the sparse \ac{THz} \ac{MIMO} channel. The \ac{BCRLB} is also determined for benchmarking the performance of the \ac{CSI} estimation techniques developed. Finally, an optimal hybrid transmit precoder and receiver combiner pair is designed, which directly relies on the beamspace domain \ac{CSI} estimates and only requires limited feedback. Finally, simulation results are provided for quantifying the improved \ac{MSE}, spectral-efficiency (SE) and \ac{BER} performance for transmission on  practical \ac{THz} \ac{MIMO} channel obtained from the \ac{HITRAN}-database.  
\end{abstract}
\IEEEpeerreviewmaketitle
\vspace{-2pt}
\section{Introduction}
\IEEEPARstart{T}{era-Hertz} (\ac{THz}) wireless systems are capable of supporting data rates up to several \ac{Tbps} \cite{han2018propagation, chen2019survey, faisal2020ultramassive} in the emerging 6G landscape. The availability of large blocks of spectrum in the \ac{THz} band, in the range of $0.1$\,\ac{THz} to $10$\,\ac{THz}, can readily fulfil the ever-increasing demand for data rates. This can in turn support several bandwidth-thirsty applications such as \ac{AR}, \ac{VR}, wireless backhaul and ultra-high speed indoor communication \cite{han2018propagation}. However, due to their high  carrier frequency, \ac{THz} signals experience severe propagation losses and blockage, beyond a few meters. Moreover, the high molecular absorption due to the  vibrations of the molecules at specific frequencies, and  the higher-order reflections \cite{jornet2011channel}  become cumbersome in the \ac{THz} band. Hence, the practical realization of \ac{THz} systems faces numerous challenges. A promising technique of overcoming these obstacles is constituted by multiple-input multiple-output (\ac{MIMO}) solutions relying on antenna arrays, which are capable of improving the signal strength at the receiver via the formation of `pencil-sharp beams' having ultra-high directional gains \cite{sarieddeen2020overview}. However, the conventional \ac{MIMO} transceiver architecture, wherein each transmit and receive antenna is connected to an individual \ac{RF} chain, becomes unsuitable at such high frequencies, mainly due to the power hungry nature of the  analog-to-digital converters coupled with their high sampling-rate \cite{he2020beamspace}. Hence, the hybrid transceiver architecture, originally proposed by Molish \textit{et al.} in their pioneering work \cite{molisch2017hybrid,zhang2005variable}, is an attractive choice for such systems, since it allows the realization of a practical transceiver employing only a few RF chains. Furthermore, in  conventional \ac{MIMO} systems, the various signal processing operations are typically implemented in the digital domain. By contrast, the signal processing tasks are judiciously partitioned between the \ac{RF} front-end and baseband processor in a hybrid \ac{MIMO} transceiver, with the former handling the analog processing via \ac{APSs}, while the latter achieves baseband processing in a \ac{DSP}. {Naturally, the overall performance of the hybrid architecture, for example, its  \ac{ASE} and bit-error-rate (\ac{BER}), critically depend on the design of the baseband and \ac{RF} precoder/ combiner, which ultimately relies on the accuracy of the available channel state information (\ac{CSI}). Thus, high-precision channel estimation holds the key for attaining robust performance and ultimately for realizing the full potential of \ac{THz} \ac{MIMO} systems. A detailed overview and a comparative survey of the related works is presented next.}

\subsection{Related Works and Contributions}
The pioneering contribution of Jornet and Akyildiz \cite{jornet2011channel} developed a novel channel model for the entire \ac{THz} band, i.e. for the band spanning $0.1-10$ \ac{THz}. Their ground-breaking work relied on the concepts of radiative transfer theory \cite{goody1995atmospheric} and molecular absorption for developing a comprehensive model  \cite{rothman2009hitran}. Their treatise evaluated the total path-loss by meticulously accounting for the molecular absorption, the reflections as well as for the free-space loss components. Later Yin and Li \cite{lin2015adaptive} developed a general \ac{MIMO} channel model for a hybrid \ac{THz} system and subsequently proposed distance-aware adaptive beamforming techniques for improving the \ac{SNR}. However, their framework assumes the availability of perfect \ac{CSI}, which is rarely possible in practice. To elaborate, \ac{CSI} estimation in a \ac{THz} hybrid \ac{MIMO} system is extremely challenging owing to the low \ac{SNR} and massive number of antenna elements. Hence, the conventional \ac{LS} and \ac{MMSE}-based \ac{CSI} estimation would incur an excessive pilot-overhead. Therefore, they are unsuitable for \ac{CSI} acquisition in practical \ac{THz} systems. 

Early solutions  \cite{9454383, alkhateeb2014channel, 9445013, el2014spatially, srivastava2019quasi, gao2019wideband, 9165794, 8732197} proposed for the \ac{mmWave} band exploited the angular-sparsity of the channel to achieve improved \ac{CSI} estimation and tracking at a substantially reduced pilot-overhead. {Several optimization and machine learning based algorithms are also proposed for hybrid transceiver design in mmWave MIMO systems. In this context, the authors of \cite{8329410} proposed joint beam selection and precoder design for maximizing the sum-rate of a downlink multiuser mmWave MIMO system under transmit power constraints. The pertinent optimization problem has been formulated as a weighted minimum mean squared error (WMMSE) problem, which is then efficiently solved using the penalty dual decomposition method. A joint hybrid precoder design procedure has been described in \cite{8606437} for full-duplex relay-aided multiuser mmWave MIMO systems, considering also the effects of imperfect CSI. The authors of \cite{9110863} and \cite{9610037} successfully developed two-timescale hybrid precoding schemes for maximizing the sum-rate, and reducing the complexity and CSI feedback overhead. A frame-based transmission scenario is considered in their work, wherein each frame comprises a fixed number of time slots. The long-timescale RF precoders are designed based on the available channel statistics and are updated once in a frame. By contrast, the short-timescale baseband precoders are optimized for each time slot based on the low-dimensional effective CSI. Hence, an optimization based solution is developed in \cite{9110863}, whereas a deep neural network (DNN)-aided technique is designed in \cite{9610037}.} The angular-sparsity is also a key feature of the \ac{THz} \ac{MIMO} channel \cite{sarieddeen2020overview,yan2019dynamic}, which arises due to the highly directional beams of large antenna arrays, coupled with high propagation losses and signal blockage in the \ac{THz} regime. In fact, Sarieddeen \textit{et al.}   \cite{sarieddeen2020overview} showed that the \ac{THz} \ac{MIMO} channel is more sparse than its \ac{mmWave} counterpart. However, there are only a few recent studies, such as \cite{schram2018compressive,schram2019approximate,ma2020joint}, which develop sparse recovery based \ac{CSI} estimation techniques for \ac{THz} \ac{MIMO} systems. A brief review of these and the gaps in the existing \ac{THz} literature are described next.  

The early work of Gao \textit{et al.} \cite{gao2016fast}  successfully developed an \textit{a priori} information aided fast \ac{CSI} tracking algorithm for \ac{DLA} array based \ac{THz} \ac{MIMO} systems. Their model relies on a practical user mobility trajectory \cite{zhou1999tracking} to develop a time-evolution based framework for the \ac{AoA}/ \ac{AoD} of each user. Subsequent contributions in this direction, such as \cite{stratidakis2019cooperative} and \cite{stratidakis2020low}, consider \ac{BS} cooperation and a  multi-resolution codebook, respectively, for improving the accuracy of channel tracking obtained via the \textit{a priori} information aided scheme of \cite{gao2016fast}. However, this improved tracking accuracy is achieved at the cost of inter-\ac{BS} cooperation, which necessitates additional infrastructure and control overheads. Kaur \textit{et al.} \cite{kaur2020enhanced} developed a model-driven deep learning technique for enhancing the channel tracking accuracy in a \ac{THz} \ac{MIMO} system. Their algorithm relies on a deep convolutional neural network trained offline in advance to learn the non-linear relationship between the estimates based on \cite{gao2016fast} and the original channel. Another impressive contribution \cite{he2020beamspace} by He \textit{et al.}  proposes a model-driven unsupervised learning network for beamspace channel estimation in wide-band \ac{THz} \ac{MIMO} systems. Furthermore, a deep learning assisted signal detection relying on single-bit quantization is proposed in the recent contribution \cite{9536661}. A fundamental limitation of \cite{gao2016fast,
stratidakis2019cooperative,
stratidakis2020low,
kaur2020enhanced} is that they consider single antenna users. More importantly, their estimation accuracy is highly sensitive to the accuracy of the time-evolution model employed and they do not incorporate the effect of molecular absorption into their THz channel, which renders the model inaccurate in reproducing the true radio propagation environment. 

{Schram \textit{et al.} \cite{schram2019approximate} employed an \ac{AMP}-based framework for \ac{CSI} estimation in \ac{THz} systems. The sparse channel estimation framework developed considers only a \ac{SISO} \ac{THz} system, where the channel impulse response (CIR) is assumed to be sparse. Ma \textit{et al.} \cite{ma2020joint} conceived sparse beamspace \ac{CSI} estimation for \ac{IRS}-based \ac{THz} \ac{MIMO} systems.} The optimal design of the phase shift matrix at the \ac{IRS} has been determined in their work based on the \ac{BS} to \ac{IRS} and \ac{IRS} to \ac{UE} \ac{THz} \ac{MIMO} channels. Recent treatises, such as \cite{dovelos2021channel, sha2021channel, balevi2021wideband}, address the problem of wideband \ac{CSI} acquisition in \ac{THz} systems. Specifically, Dovelos \textit{et al.} \cite{dovelos2021channel} consider an \ac{OFDM}-based \ac{THz} hybrid \ac{MIMO} system and develop orthogonal matching pursuit (\ac{OMP})-based techniques for \ac{CSI} estimation. Balevi and Andrews  \cite{balevi2021wideband} have considered generative adversarial networks for channel estimation in an \ac{OFDM}-based \ac{THz} hybrid \ac{MIMO} system. On the other hand, Sha and Wang \cite{sha2021channel} derived a \ac{CSI} estimation and equalization technique for a single-carrier THz \ac{SISO} system accounting also for realistic \ac{RF} impairments. A list of novel contributions of our paper is presented next. Our novel contributions are also boldly and explicitly contrasted to the existing literature in Table-\ref{tab:lit_rev}.

\begin{table*}

    \centering
\caption{\small Boldly contrasting our novel contributions to the state-of-the-art}\label{tab:lit_rev}

\begin{tabular}{|l|c|c|c|c|c|c|c|c|c|c|c|c|c|c|c|c|c|c|}

    \hline

\textbf{Feautures}  &\cite{jornet2011channel} &\cite{he2020beamspace}  &\cite{lin2015adaptive}  &\cite{ma2020joint} &\cite{kaur2020enhanced}  &\cite{gao2016fast} &\cite{schram2019approximate} & \cite{balevi2021wideband} &\cite{dovelos2021channel} &\cite{sha2021channel} &\textbf{Proposed} \\

 \hline

\ac{THz} hybrid \ac{MIMO}

& & &\checkmark & \checkmark & & & & \checkmark &\checkmark & & \checkmark\\

 \hline
 
\ac{APSs}-based hybrid architecture

& & &\checkmark & \checkmark & & & & \checkmark &\checkmark & & \checkmark\\

 \hline
 
Single antenna users

 & &\checkmark &\checkmark & &  \checkmark &\checkmark & \checkmark & &\checkmark & & \checkmark\\

 \hline
 
\ac{CSI} estimation 
& &\checkmark & & \checkmark &  \checkmark &\checkmark &\checkmark & \checkmark &\checkmark & \checkmark & \checkmark\\

 \hline
 
 Diffused-ray modeling & & & & & & & & & &  &\checkmark\\

\hline
 
 Angular-sparsity 
 & &\checkmark &\checkmark &\checkmark & &\checkmark & \checkmark & &\checkmark & & \checkmark\\

 \hline
 
Molecular absorption losses &\checkmark & &\checkmark & &  & & & &\checkmark & & \checkmark\\

\hline
 
Reflection losses

&  &  &\checkmark &  &   &    &  &   &\checkmark &   & \checkmark\\

 \hline
 
 Optimal pilot design

&  &  &  &  &   &  &  &   &  & \checkmark & \checkmark\\

 \hline
 
Transceiver design

&  &  &  &\checkmark &   &    &  &   &\checkmark &   & \checkmark\\

 \hline
 
Optimal power allocation

&  &  &\checkmark &  &   &    &  &   &  &  & \checkmark\\
 \hline

Limited \ac{CSI} feedback

&  &  &  &  &   &    &  &   &  &   & \checkmark\\

 \hline
 
MSE lower bound

&  &  &  &  &   &    &  &   & \checkmark & \checkmark & \checkmark\\
\hline
\end{tabular}
\end{table*}

{
\subsection{Novel Contributions}
\begin{enumerate}
    \item We commence by developing a practical distance and frequency dependent \ac{THz} \ac{MIMO} channel model that also incorporates the molecular absorption and reflection losses together with the traditional free-space loss. Note that almost all the existing contributions utilize the classical Saleh-Valenzuela channel model of \cite{lin2015indoor}, which does not consider the diffused rays for each multipath component together with 
    first- and second-order reflections. Furthermore, the path-gains in most of the existing treatises are simply modeled as Rayleigh fading channel coefficients without considering the molecular absorption and multiple reflections. Hence, an important aspect of the channel model developed is that it incorporates several diffused rays for each of the reflected multipath components including their associated reflection and molecular absorption losses. This results in broadening of the beamwidths of the signals and mimics a practical THz MIMO channel.  
    \item The existing research on the development of sparse CSI estimation schemes for a point-to-point analog phase shifter (APS) based hybrid MIMO THz system is very limited, since most of them have considered only single-antenna users, focusing predominantly on discrete lens antenna (DLA) arrays. Hence for  considering a point-to-point APS-based hybrid MIMO architecture, an efficient frame-based channel estimation model is developed, which frugally employs a low number of pilot beams for exciting the various angular modes of the channel. Subsequently, using a suitable `sparsifying'-dictionary, a beamspace representation is developed for the \ac{THz} \ac{MIMO} channel, followed by the pertinent sparse channel estimation model. For this, \ac{BL}-based channel estimation techniques are derived for exploiting the sparsity of the \ac{THz} \ac{MIMO} channel. Note that the proposed BL-based technique is novel in the context of THz MIMO channel estimation, since it has not been explored as yet in THz hybrid MIMO systems. 
    \item The design of the optimal pilot beams used for CSI estimation, which can significantly enhance sparse signal recovery, has not been considered in the existing THz literature either. Moreover, it is also desirable to develop bounds to benchmark the performance of the CSI estimation schemes. To this end, another key contribution of this work is the design of a specific pilot matrix that minimizes the so-called `total-coherence'\footnote{The total coherence of a matrix $\widetilde{\mathbf \Phi}$ having $G$ columns, denoted as $\mu^t\left(\widetilde{\mathbf \Phi}\right)$, is defined as $\mu^t\left(\widetilde{\mathbf \Phi}\right) = \sum_{i = 1}^{G}\sum_{j=1, j \neq i}^{G} \left\vert \widetilde{\mathbf \Phi}_{i}^H  \widetilde{\mathbf \Phi}_{j} \right\vert^2$, where the quantities $\widetilde{\mathbf \Phi}_{i}$ and $\widetilde{\mathbf \Phi}_{j}$ represent the $i$th and $j$th columns, respectively, of the matrix $\widetilde{\mathbf \Phi}$.} defined in \cite{elad2007optimized, li2013projection} for enhancing the performance of sparse signal recovery. Furthermore, to benchmark the \ac{MSE} performance of our sparse CSI estimators, the Bayesian Cramer-Rao lower bound (\ac{BCRLB}) is also derived for the \ac{CSI} estimates.
    \item To the best of our knowledge, the existing hybrid transceiver design approaches found in the \ac{THz} literature, such as \cite{yuan2018hybrid}, assume the availability of perfect \ac{CSI}, which is impractical due to the large number of antennas, resulting in excessive pilot overheads. Crucially, no joint beamspace channel estimation and hybrid transceiver design procedure is available in the \ac{THz} literature. To address this problem, a capacity-approaching hybrid transmit precoder (TPC) and \ac{MMSE}-optimal hybrid receiver combiner (RC) are developed, which can directly employ the estimate of the beamspace domain channel obtained from the proposed \ac{CSI} estimators. The proposed algorithm requires only limited CSI of the beamspace channel, namely the non-zero coefficients and their respective indices, which substantially reduces the feedback required. Furthermore, in contrast to the existing hybrid transceiver designs \cite{alkhateeb2014channel, el2014spatially, 9179621, 9395091}, the proposed hybrid transceiver design requires no iterations, and hence it is computationally efficient.
    \item Our simulation results demonstrate the enhanced  performance of our channel estimators, TPC and RC for various practical simulation parameters. In this context, this paper calculates the molecular absorption coefficient using the parameters obtained from the HITRAN database \cite{hitran}, which is suitable for the entire THz band, specifically for the higher end spanning $1$ to $10$ THz. On the other hand, most of the existing works employ  models, which are only valid for the lower end around $0.1$ to $0.3$ THz.
\end{enumerate}
}
\subsection{Organization and Notation}
{The main focus of this work is on hybrid transceiver design relying on the BL-based estimated beamspace domain CSI. To achieve this, in Section-\ref{sec:THz_channel_model}, we begin with the THz MIMO system and channel model, which incorporates the specific molecular absorption and reflection losses arising in the THz regime. This is followed by developing its sparse beamspace domain representation and a novel frame-based channel estimation model in Section-\ref{sec:channel_est}, which excites various angular modes of the THz MIMO channel. Furthermore, in order to improve the sparse CSI estimation performance, the \textit{mutual coherence} of the equivalent sensing matrix has also been minimized in Section-\ref{sec:channel_est}, which results in the optimal choice of the training precoders/combiners to be employed during channel estimation. Subsequently, the proposed BL and MBL-based sparse channel estimation schemes are developed in Section-\ref{sec:BL_based_est}, which is followed by the BCRLB for benchmarking their CSI estimation performance. Finally, based on the estimated CSI, the problem of designing the capacity-optimal hybrid precoder and MMSE-optimal hybrid combiner is addressed in Section-\ref{prec_comb}. Our simulation results are presented in Section-\ref{sec:simulation_results}, followed by our conclusions in Section-\ref{sec:conclusion}.}

Notation: 
The notation $\mathrm{floor}[a]$ represents the greatest integer, which is less than $a$, whereas $\mathrm{rem}[a,b]$ denotes the remainder, when $a$ is divided by $b$; the $i$th element of the vector $\mathbf{a}$ and  $(i,j)$th element of the matrix $\mathbf{A}$ are denoted by $\mathbf{a}(i)$ and $\mathbf{A}(i,j)$, respectively; $\mathbf{I}_N$ denotes an identity matrix of size $N$; $\mathrm{vec}(\mathbf{A})$ vectorizes the columns of the matrix $\mathbf{A}$ and $\mathrm{vec}^{-1} (\mathbf{a})$ denotes the inverse vectorization operation; the Kronecker product of two matrices $\mathbf{A}$ and $\mathbf{B}$ is denoted by $\mathbf{A} \otimes \mathbf{B}$;
the $l_2$- and Frobenius-norm are represented by $\|\cdot\|_2$  and $\|\cdot\|_F$, respectively. 
\begin{figure*}
\centering\includegraphics[scale=0.6]{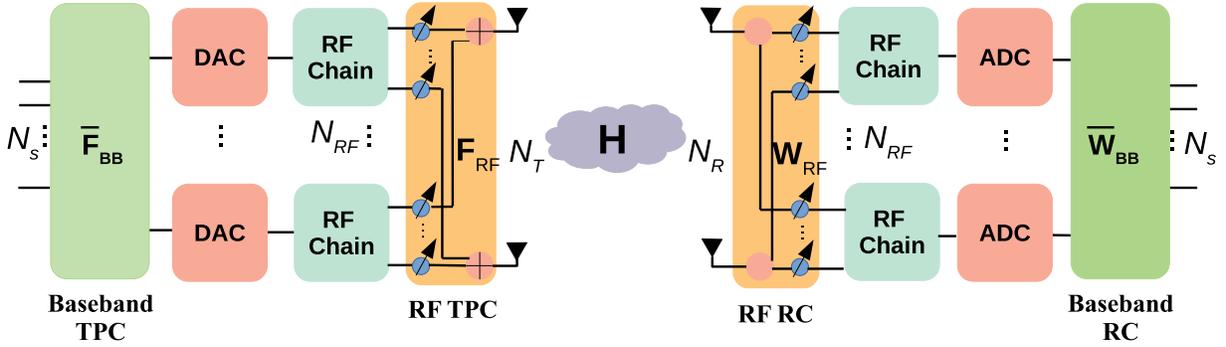}
\centering\caption{Block diagram of a \ac{THz} hybrid \ac{MIMO} system. \vspace{-15pt}}
\centering\label{fig:THz_mimo:schematic}
\end{figure*}
\section{\ac{THz} \ac{MIMO} System and Channel Model}\label{sec:THz_channel_model}
The schematic of our \ac{THz} hybrid \ac{MIMO} system is portrayed in Fig. \ref{fig:THz_mimo:schematic}, where $N_T$ and $N_R$ denote the number of transmit antennas (TAs) and receive antennas (RAs), respectively, whereas $N_{\text{RF}}$ denotes the number of \ac{RF} chains. Furthermore, $N_S$ is the number of data streams, where $N_S \leq N_{\text{RF}}$, while $N_{\text{RF}} << \min (N_T, N_R)$ \cite{ma2020joint,sarieddeen2020overview}. The transmitter is composed of two major blocks, the digital baseband TPC $\bar{\mathbf{F}}_{\text{BB}} \in \mathbb{C}^{N_{\text{RF}} \times N_S}$ and the analog \ac{RF} TPC $\bar{\mathbf{F}}_{\text{RF}} \in \mathbb{C}^{N_{T} \times N_{\text{RF}}}$. At the receiver side, $\bar{\mathbf{W}}_{\text{RF}} \in \mathbb{C}^{N_R \times N_{\text{RF}}}$ denotes the \ac{RF} RC, whereas  $\bar{\mathbf{W}}_{\text{BB}} \in \mathbb{C}^{N_{\text{RF}} \times N_S}$ represents the baseband RC. As described in \cite{ma2020joint,sarieddeen2020overview}, the analog \ac{RF} TPC $\bar{\mathbf{F}}_{\text{RF}}$ and RC $\bar{\mathbf{W}}_{\text{RF}}$ are comprised of \ac{APSs}. Hence, for simplicity, these are constrained as $\vert \bar{\mathbf{F}}_{\text{RF}}(i,j)\vert = \frac{1}{\sqrt{N_T}}, \vert \bar{\mathbf{W}}_{\text{RF}}(i,j)\vert = \frac{1}{\sqrt{N_R}}, \forall i,j$. Thus, the baseband system model of our \ac{THz} \ac{MIMO} system is given by 
\begin{align}
\bar{\mathbf y} = \bar{\mathbf{W}}_{\text{BB}}^H \bar{\mathbf{W}}_{\text{RF}}^H \mathbf{H} \bar{\mathbf{F}}_{\text{RF}} \bar{\mathbf{F}}_{\text{BB}} \bar{\mathbf x} + \bar{\mathbf{W}}_{\text{BB}}^H \bar{\mathbf{W}}_{\text{RF}}^H \bar{\mathbf v}, \label{eq:th_sys_mdl}
\end{align}
where  $\bar{\mathbf y} \in \mathbb{C}^{N_S \times 1}$ is the signal vector received at the output of the baseband RC, $\bar{\mathbf x} \in \mathbb{C}^{N_S \times 1}$ represents the transmit baseband signal vector at the input of the baseband TPC, whereas the quantity $\bar{\mathbf v} \in \mathbb{C}^{N_R \times 1}$ is the complex \ac{AWGN} at the receiver having the distribution of $\mathcal{CN}\left(\mathbf{0}_{N_R\times 1}, \sigma_v^2 \textbf{I}_{N_R}\right)$. The matrix $\mathbf{H} \in \mathbb{C}^{N_R \times N_T}$ in \eqref{eq:th_sys_mdl} represents the baseband equivalent of the \ac{THz} \ac{MIMO} channel, whose relevant model is described next.
\subsection{\ac{THz} \ac{MIMO} Channel Model}
As described in \cite{lin2015adaptive}, the \ac{THz} \ac{MIMO} channel can be modeled as the aggregation of a \ac{LoS} and a few \ac{NLoS} components. The \ac{LoS} propagation results in a direct path between the \ac{BS} and the \ac{UE}, whereas the \ac{NLoS} propagation results in some indirect multipath rays after reflection from the various scatterers present in the environment. Thus, the \ac{THz} \ac{MIMO} channel $\mathbf{H}$, which is  a function of the operating frequency $f$ and distance $d$, can be expressed as 
\begin{align}
\mathbf{H}(f,d)=\mathbf{H_\text{LoS}}(f,d)+\mathbf{H_\text{NLoS}}(f,d),
\label{eq:ch_model}
\end{align}
where the \ac{LoS} and \ac{NLoS} components are given by
\begin{align}
    \mathbf{H_\text{LoS}}(f,d)&=\sqrt{\frac{N_{T}N_{R}}{N_{\text{ray}}}} \sum_{j=1}^{N_{\text{ray}}}\alpha_{\text{L},j}(f,d) G_{t}^{a} G_{r}^{a}\mathbf{a}_{r}\left(\phi_{\text{L},j}^{r}\right)\mathbf{a}_{t}^{H}\left(\phi_{\text{L},j}^{t}\right), \label{eq:LOS_model}\\
    \mathbf{H_\text{NLoS}}(f,d)&=\sqrt{\frac{N_{T}N_{R}}{N_{\text{NLoS}}N_{\text{ray}}}} \sum\limits_{i=1}^{N_\text{NLoS}} \sum_{j=1}^{N_{\text{ray}}} \alpha_{i,j}(f,d) G_{t}^{a}G_{r}^{a}\mathbf{a}_{r} \left(\phi_{i,j}^{r}\right)\mathbf{a}_{t}^{H}\left(\phi_{i,j}^{t}\right)\label{eq:NLOS_model}.
\end{align}
Here, the quantities $\alpha_{\text{L},j}(f,d)$ and  $\alpha_{i,j}(f,d)$ represent the complex-valued path-gains of the \ac{LoS} and \ac{NLoS} components, respectively, $N_\text{NLoS}$ denotes the number of \ac{NLoS} multipath components, whereas $N_{\text{ray}}$ signifies the number of diffused-rays in each multipath component. Furthermore, $G_t^{a}$ and $G_r^{a}$ represent the TA and RA gains, respectively. The quantities $\phi^{r}_{\text{L},j}$ and $\phi^{t}_{\text{L},j}$ denote the \ac{AoA} and \ac{AoD} of the $j$th ray in the \ac{LoS} multipath component, respectively, whereas $\phi^{r}_{i,j}$ and $\phi^{t}_{i,j}$ represent the \ac{AoA} and \ac{AoD} of the $j$th diffuse-ray in the $i$th \ac{NLoS} multipath component. The vectors $\mathbf{a}_{r}(\phi^{r})  \in \mathbb{C}^{N_R \times 1}$ and $\mathbf{a}_{t}(\phi^{t})  \in \mathbb{C}^{N_T \times 1}$ denote the array response vectors of the \ac{ULA} corresponding to the \ac{AoA} $\phi^{r}$ at the receiver and \ac{AoD} $\phi^{t}$ at the transmitter, respectively. These are defined as
\begin{align}
\mathbf{a}_r(\phi^r)=\displaystyle \frac{1}{\sqrt{N_R}}\Big [ 1,  e^{-j\frac{2 \pi }{\lambda }d_r\cos \left(\phi^r\right)}, \ldots, e^{-j\frac{2\pi }{\lambda } {(N_R-1)}d_r\cos \left( \phi^r\right)}\Big ]^T,\\
\mathbf{a}_t(\phi^t)=\displaystyle \frac{1}{\sqrt{N_T}} \Big [ 1,  e^{-j\frac{2 \pi }{\lambda }d_t\cos \left( \phi^t \right)}, \ldots, e^{-j\frac{2\pi }{\lambda } {(N_T-1)}d_t\cos \left( \phi^t \right)}\Big ]^T,\label{eq:Array_resp_vector}
\end{align}
where $d_r$ and $d_t$ represent the antenna-spacings at the receiver and transmitter, respectively, and $\lambda$ denotes the operating wavelength. 

Let the complex path-gain $\alpha(f,d)$ be expressed as  $\alpha(f,d)=|\alpha(f,d)|e^{j\psi}$, where $|\alpha(f,d)|$ is the magnitude of the complex path-gain and $\psi$ is the associated independent phase shift. According to \cite{lin2015adaptive,jornet2011channel}, the magnitude of the \ac{LoS} path gain $|\alpha_{\text{L},j}(f,d)|$ can be modeled as 
\begin{align}
|\alpha_{\text{L},j}(f,d)|^{2}=L_\text{spread}(f,d)L_\text{abs}(f,d),\label{eq:Los_pg}
\end{align} 
where $L_\text{spread}(f,d)$ and $L_\text{abs}(f,d)$ represent the spreading (or the free-space) and molecular absorption losses respectively, which are given by
\begin{align}
L_\text{abs}(f,d)=e^{-k_\text{abs}(f)d},\ \ L_\text{spread}(f,d)=\left(\frac{c}{4\pi fd}\right)^{2}.\label{eq:losses}
\end{align}
Here, $c$ denotes the speed of light in vacuum and $k_\text{abs}(f)$ is the molecular absorption coefficient. 
Similarly, for the $j$th diffuse-ray of the $i$th \ac{NLoS} multipath component, the magnitude of the complex path-gain can be expressed as \cite{lin2015adaptive,sarieddeen2020overview} 
\begin{align}
\left|\alpha_{i,j}(f,d)\right|^{2}=\Gamma_{i,j}^{2}(f)L_\text{spread}(f,d)L_\text{abs}(f,d),\label{eq:Nlos_pg}
\end{align}
where $\Gamma_{i,j}(f)$ denotes the first-order reflection coefficient of the $j$th diffuse-ray of the $i$th \ac{NLoS} component. For higher-order reflections, the equivalent reflection coefficient is equal to the product of individual reflection coefficients of the respective scattering media. Further details on the calculation of the absorption coefficient $k_\text{abs}(f)$ and the reflection coefficient $\Gamma_{i,j}(f)$ are given in the subsequent subsections.
\subsection{Calculation of the Reflection Coefficient $\Gamma(f)$}
Due to the small wavelength of the \ac{THz} signal, the reflection coefficient $\Gamma(f)$ is an important parameter to be taken into account while evaluating the losses of the \ac{NLoS} components \cite{piesiewicz2007scattering,lin2015adaptive}. This is in turn defined in terms of the Fresnel reflection coefficient $(\gamma)$ and the Rayleigh roughness factor $(\varrho)$, as $\Gamma(f)=\gamma(f)\varrho(f),$ where the coefficients $\gamma(f)$ and $\varrho(f)$ are given by:
\begin{align}
\gamma(f)=\frac{Z(f)\cos\left(\theta_\text{in}\right)-Z_{0}\cos\left(\theta_\text{ref}\right)}{Z(f)\cos\left(\theta_\text{in}\right)+Z_{0}\cos\left(\theta_\text{ref}\right)} \text{ \ and \ } 
\varrho(f)=e^{-\tfrac{1}{2}\left(\tfrac{4\pi f\sigma\cos\left(\theta_\text{in}\right)}{c}\right)^{2}}. \label{eq:rayleigh_factor}
\end{align}
In the above expressions, $\theta_\text{in}$ denotes the angle of incidence, while $\theta_\text{ref}$ represents the angle of refraction, which obeys   $\theta_\text{ref}=\sin^{-1}\left(\sin(\theta_\text{in})\frac{Z(f)}{Z_{0}}\right)$. The quantity $Z(f)$ denotes the wave impedance of the reflecting medium, whereas $Z_{0}=377\,\Omega$ represents the wave impedance of the free space and $\sigma$ in \eqref{eq:rayleigh_factor} denotes the standard deviation of the reflecting surface's roughness.

\subsection{Absorption Coefficient $k_{\text{abs}}(f)$ \cite{jornet2011channel}}\label{sec:k_abs}
As described in \cite{jornet2011channel}, the absorption coefficient $k_{\text{abs}}(f)$ of the propagation medium at frequency $f$ can be evaluated as
\begin{align}
    k_{\text{abs}}(f)=\sum_{i,g}k_{\text{abs}}^{i,g}(f),
\end{align}
where, $k_{\text{abs}}^{i,g}$ denotes the absorption coefficient of the $i$th isotopologue\footnote{Molecules, which only differ from others in their isotopic composition, are termed as isotopologues of each other.} of the $g$th gas. The quantity $k_{\text{abs}}^{i,g}(f)$ can be mathematically defined as
\begin{align}
    k_{\text{abs}}^{i,g}(f)=\left(\frac{p}{p_{0}}\right)\left(\frac{T_\text{STP}}{T}\right)Q^{i,g}\sigma^{i,g}(f),
\end{align}
where $p$ and $T$ denote the system pressure and temperature, respectively, while $T_\text{STP}$ and $p_0$ represent the temperature at standard pressure and reference pressure, respectively. The quantity $\sigma^{i,g}$ denotes the absorption cross-section of the $i$th isotopologue of the $g$th gas, defined as $\sigma^{i,g}(f)=S^{i,g}G^{i,g}(f),$ and $Q^{i,g}$ is the molecular volumetric density, defined as $Q^{i,g}={\left(\frac{p}{RT}\right)}q^{i,g}N_{A}.$ Here, $R$ denotes the gas constant and $N_A$ represents Avogadro's number. The quantities $q^{i,g}$ and $S^{i,g}$ signify the mixing ratio and the line intensity, respectively, of the $i$th isotopologue of the $g$th gas, which can be directly obtained from the \ac{HITRAN} database  \cite{rothman2009hitran}. The quantity $G^{i,g}(f)$ is the spectral line shape, defined as
\begin{align}
    G^{i,g}(f)=\left(\frac{f}{f_{c}^{i,g}}\right){\frac{\tanh\left({\frac{cfh}{2k_{B}T}}\right)}{\tanh\left({\frac{cf_{c}^{i,g}h}{2k_{B}T}}\right)}}F^{i,g}(f), 
\end{align}
where $k_B$ denotes the Boltzmann constant, $h$ represents the Planck constant and $F^{i,g}(f)$ is the Van Vleck-Weisskopf line shape \cite{van1945shape}, which is evaluated as follows
\begin{align}
&F^{i,g}(f)=\frac{100cf\alpha_L^{i,g}}{\pi f_{c}^{i,g}}\sum_{n=1}^2\frac{1}{(f+(-1)^nf_{c}^{i,g})^{2}+(\alpha_L^{i,g})^{2}}.\!\!
\end{align}
The quantities $f_{c}^{i,g}$ and $\alpha_L^{i,g}$ obey:
\begin{align}
f_{c}^{i,g} = f_{c0}^{i,g}+\delta^{i,g}\frac{p}{p_{0}} \text{ and } \alpha_L^{i,g} = \left[(1-q^{i,g})\alpha_{0}^{\text{air}}+q^{i,g}\alpha_{0}^{i,g}\right] \left(\frac{p}{p_0}\right)\left(\frac{T_0}{T}\right)^{\gamma},
\label{eq:half_width}
\end{align}
where $f_{c0}^{i,g}$ and $\delta^{i,g}$ denote the zero-pressure resonance frequency and linear pressure shift, respectively, which are also obtained from the \ac{HITRAN} database. In \eqref{eq:half_width}, $\alpha_{0}^{\text{air}}$ and $\alpha_{0}^{i,g}$ represent the broadening coefficient of the air and of the $i$th isotopologue of the $g$th gas, respectively, whereas $\gamma$ denotes the temperature broadening coefficient, all of which can be directly obtained from the \ac{HITRAN} database, and the quantity $T_0$ denotes the reference temperature. The parameters involved in the calculation of molecular absorption coefficient $k_{\text{abs}}(f)$, their units and the values of various constants are summarized in Table-I of \cite{jornet2011channel}. Furthermore, the \ac{HITRAN} database is accessible online from \cite{hitran}. 

{From the channel model presented in this section, it can be readily observed that the THz MIMO channel is significantly different from its mmWave counterpart. First of all, note that the THz MIMO channel is highly dependent on the carrier frequency $f$ and distance $d$, not just due to the free-space loss $L_\text{spread}(f,d)$, but more importantly due to the nature of molecular absorption loss $L_\text{abs}(f,d) = e^{-k_\text{abs}(f)d}$, where the absorption coefficient $k_\text{abs}(f)$ is highly dependent on the molecular composition of the propagation medium, system pressure, temperature and the operating frequency. In fact, as evaluated in \cite{jornet2011channel}, even the water vapor molecules present in a standard medium lead to a significant loss, which affects the overall system performance in the THz band. By contrast, in the mmWave band, these atmospheric losses only become significant in the presence of raindrops/ fog. Due to this, THz signals experience severe propagation losses beyond a few meters. Furthermore, due to the extremely short wavelength of THz signals, the indoor surfaces, which can be regarded as smooth in the comparatively lower mmWave band, now appear rough in the THz regime \cite{lin2015adaptive}. Hence, it can be observed from \eqref{eq:Los_pg} and \eqref{eq:Nlos_pg} that the complex-valued path gains $\alpha_{\text{L},j}(f,d)$ and $\alpha_{i,j}(f,d)$ in the THz band for the LoS and NLoS components, respectively, differ significantly due to their increased higher-order reflection losses. Additionally, due to the large number of antennas, the THz MIMO channel becomes highly directional and more sparse in nature in comparison to its mmWave counterpart. It has also been verified in \cite{jornet2011channel} and also in our simulation results that at certain frequencies, the molecular absorption is very high, which reduces the total bandwidth to just a few transmission windows. Hence, the molecular absorption plays a critical role in deciding the operating frequency and bandwidth.} The next section describes the channel estimation model proposed for our \ac{THz} \ac{MIMO} systems. For ease of notation, we drop the quantities $(f,d)$ from the \ac{THz} \ac{MIMO} channel representation in the subsequent sections, since these parameters are fixed for the channel under consideration.
\section{\ac{THz} \ac{MIMO} Channel Estimation Model}
\label{sec:channel_est}
Consider the transmission of $N_F = \frac{N_T}{N_\text{RF}}$ training frames and $M_T$ training vectors, where $M_T<N_T$. This implies that $\frac{M_T}{N_F}$ training vectors are  transmitted in each frame. Let $\mathbf{F}_{\text{RF},i} \in \mathbb C^{N_T \times N_\text{RF}}$ represent the RF training TPC and $\mathbf{X}_{p,i} \in \mathbb C^{N_\text{RF} \times \frac{M_T}{N_F}}$ denote the pilot matrix corresponding to the $i$th training frame. The received pilot matrix $\widetilde{\mathbf{Y}}_{i} \in \mathbb C^{N_R \times \frac{M_T}{N_F}}$ can be represented as   
\begin{align}
\widetilde{\mathbf{Y}}_{i} =  \mathbf{H} \mathbf{F}_{\text{RF},i} {\mathbf{X}_{p,i}} + \widetilde{\mathbf{V}}_{i}, 
\end{align}
where  $\widetilde{\mathbf{V}}_{i} \in \mathbb C^{N_R \times \frac{M_T}{N_F}}$ denotes the noise matrix having  \ac{i.i.d.} elements obeying $\mathcal{CN}(0,\sigma_v^2)$. Upon concatenating   $\widetilde{\mathbf{Y}}_{i}$  for $1 \leq i \leq N_F$, as $\widetilde{\mathbf{Y}} =\big[\widetilde{\mathbf{Y}}_{1}, \widetilde{\mathbf{Y}}_{2}, \cdots, \\ \widetilde{\mathbf{Y}}_{N_F}\big] \in \mathbb C^{N_R \times M_T}$, one can model the received pilot matrix as
\begin{align}
\widetilde{\mathbf{Y}} &= \mathbf{H} \mathbf{F}_{\text{RF}} \mathbf{X}_{p} + \widetilde{\mathbf{V}},
\end{align}
where the various quantities are defined as $\mathbf{F}_{\text{RF}} = \left[\mathbf{F}_{\text{RF},1}, \mathbf{F}_{\text{RF},2},\cdots,\mathbf{F}_{\text{RF},N_F}\right] \in  \mathbb C^{N_T \times N_T}, \widetilde{\mathbf{V}} = \left[\widetilde{\mathbf{V}}_{1}, \widetilde{\mathbf{V}}_{2}, \cdots, \widetilde{\mathbf{V}}_{N_F}\right] \in \mathbb C^{N_R \times M_T}$ and 
\begin{align}
\mathbf{X}_{p} = \mathrm{blkdiag} \left(\mathbf{X}_{p,1}, \mathbf{X}_{p,2}, \cdots, \mathbf{X}_{p,N_F}\right) \in \mathbb C^{N_T \times M_T}, \label{eq:X_p}
\end{align}
Similarly, let $N_C = \frac{N_R}{N_\text{RF}}$ represent the number of combining-steps, whereas $M_R$ denote the number of combining vectors. In each combining-step, we combine the pilot output $\widetilde{\mathbf{Y}}$ using $\frac{M_R}{N_C}$ combining vectors in the baseband, where $M_R<N_R$. Let $\mathbf{W}_{\text{RF},j} \in \mathbb{C}^{N_R \times N_\text{RF}}$ denote the RF RC and $\mathbf{W}_{\text{BB},j} \in \mathbb{C}^{N_\text{RF} \times \frac{M_R}{N_C}}$ represent the baseband RC of the $j$th combining step. The received pilot matrix $\mathbf{Y}_j \in \mathbb C^{\frac{M_R}{N_C} \times M_T}$ at the output of the $j$th  baseband RC is obtained as $\mathbf{Y}_j = \mathbf{W}_{\text{BB},j}^H\mathbf{W}_{\text{RF},j}^H \widetilde{\mathbf{Y}}.$ Let $\mathbf{Y} = \left[ \mathbf{Y}_1^T, \mathbf{Y}_2^T, \cdots, \mathbf{Y}_{N_C}^T \right]^T \in \mathbb{C}^{M_R \times M_T}$ represent the stacked received pilot matrices $\mathbf{Y}_j, 1 \leq j \leq N_C$. The end-to-end model can be succinctly represented as
\begin{align}
\mathbf{Y}=\mathbf{W}^H_{\text{BB}} \mathbf{W}_{\text{RF}}^H \mathbf{H} \mathbf{F}_{\text{RF}} \mathbf{X}_{p}+\mathbf{V}, \label{eq:Rx_pilot_Mx_0}
\end{align}
where the various quantities have the following expressions: $\mathbf{W}_{\text{RF}} = \left[\mathbf{W}_{\text{RF},1},\mathbf{W}_{\text{RF},2}, \cdots,\mathbf{W}_{\text{RF},N_C}\right] \in  \mathbb C^{N_R \times N_R}, \mathbf{V} = \mathbf{W}^H_{\text{BB}} \mathbf{W}_{\text{RF}}^H \widetilde{\mathbf{V}} \in \mathbb{C}^{M_R \times M_T}$ and 
\begin{align}
\mathbf{W}_{\text{BB}} = \mathrm{blkdiag} \left(\mathbf{W}_{\text{BB},1},  \cdots, \mathbf{W}_{\text{BB},N_C}\right) \in \mathbb{C}^{N_R \times M_R}. \label{eq:BB_comb}
\end{align}
One can now exploit the properties of the matrix Kronecker product \cite{zhang2013kronecker} to arrive at the following \ac{THz} \ac{MIMO} channel estimation model
\begin{align}
\mathbf{y} = \mathbf{\Phi} \mathbf{h} + \mathbf{v}, \label{eq: equivalent_system_model}
\end{align}	
where $\mathbf{y} = \mathrm{vec}\left(\mathbf{Y}  \right) \in \mathbb{C}^{M_TM_R \times 1}$ represents the received pilot vector and $\mathbf{v} = \mathrm{vec}\left(\mathbf{V}  \right) \in \mathbb{C}^{M_TM_R \times 1}$ denotes the noise vector. The quantity $\mathbf{h} = \mathrm{vec}{\left(\mathbf H\right)}\in\mathbb{C}^{N_TN_R \times 1}$ is the equivalent \ac{THz} \ac{MIMO} channel vector and the matrix $\mathbf{\Phi} \in\mathbb{C}^{{M_TM_R \times N_TN_R}}$ represents the \textit{sensing matrix} obeying $\mathbf{\Phi} = \left[\left(\mathbf{X}_{p}^T \mathbf{F}_{\text{RF}}^T\right) \boldsymbol \otimes \left( \mathbf W^H_{\text{BB}} \mathbf W^H_{\text{RF}} \right)\right].$ Finally, the noise covariance matrix $\mathbf{R}_{v} \in \mathbb{C}^{M_TM_R \times M_TM_R}$, defined as $\mathbf{R}_{v} = \mathbb{E}\left\{\mathbf{v}\mathbf{v}^H\right\}$, is given as $\mathbf{R}_{v} = \sigma_v^2 \left[\mathbf{I}_{M_T} \boldsymbol \otimes \left( \mathbf{W}^H_{\text{BB}} \mathbf{W}_{\text{RF}}^H  \mathbf{W}_{\text{RF}} \mathbf{W}_{\text{BB}} \right)\right].$ At this point, it can be noted that the expressions for the conventional \ac{LS} and \ac{MMSE} estimates of the \ac{THz} \ac{MIMO} channel vector $\mathbf{h}$ can be readily  derived from the simplified model in \eqref{eq: equivalent_system_model}, as 
\begin{align}
\widehat{\mathbf{h}}^{\text{LS}} = \left(   \mathbf{\Phi} \right)^{\dagger}  \mathbf{y} \text{ \ and \ }
\widehat{\mathbf{h}}^{\text{MMSE}} = \left( \mathbf{R}_{h}^{-1} +  \mathbf{\Phi}^H \mathbf{R}_{v}^{-1}  \mathbf{\Phi} \right)^{-1}  \mathbf{\Phi}^H \mathbf{y},
\end{align}
where $\mathbf{R}_{h} = \mathbb{E} \left[ \mathbf{h} \mathbf{h}^H \right] \in\mathbb{C}^{N_TN_R \times N_TN_R}$ represents the channel's covariance matrix. However, a significant drawback of these conventional estimation techniques is that they require an over-determined system, i.e., $M_TM_R \geq N_TN_R$, for reliable channel estimation. This results in unsustainably high training overheads due to the high number of antennas. Thus, conventional channel estimation techniques are inefficient for such systems. Furthermore, as described in \cite{sarieddeen2020overview, lin2015adaptive}, the \ac{THz} \ac{MIMO} channel exhibits \textit{angular-sparsity}, which is not exploited by these conventional techniques. Leveraging the sparsity of the \ac{THz} \ac{MIMO} channel can lead to significantly improved channel estimation performance as well as bandwidth-efficiency, specifically where we have the `ill-posed' \ac{THz} \ac{MIMO} channel estimation scenario of $M_TM_R << N_TN_R$. Thus, the next subsection derives a sparse channel estimation model for \ac{THz} \ac{MIMO} systems. 
\subsection{Sparse \ac{THz} \ac{MIMO} Channel Estimation Model} 
Let $G_T$ and $G_R$ signify the angular grid-sizes obeying $(G_T, G_R) \geq \max(N_T,N_R)$. The angular grids $\Phi_{T}$ and $\Phi_{R}$ for the \ac{AoD} and \ac{AoA}, respectively, are given as follows, which are constructed by assuming the directional-cosines $\cos(\phi_i)$ to be uniformly spaced between $-1$ to $1$:
\begin{align}
\Phi_T &= \left\{ \phi_i:\cos(\phi_i) =\frac{2}{G_T}(i-1)-1, \ 1 \leq i \leq G_T\right\}\!{,}\!\! \label{eq:dir_cos_tx}\\
\Phi_R &= \left\{ \phi_j:\cos(\phi_j) =\frac{2}{G_R}(j-1){-}1, \ 1 \leq j \leq G_R\right\}\!{.}\!\! \label{eq:dir_cos_rx} 
\end{align}
Let $\mathbf{A}_R (\Phi_R) \in \mathbb C^{N_R \times G_R}$ and $\mathbf{A}_T (\Phi_T ) \in  \mathbb C^{N_T \times G_T}$ represent the dictionary matrices of the array responses constructed using the angular-grids $\Phi_{R}$ and $\Phi_{T}$ as follows
\begin{align}
\mathbf A_R (\Phi_R ) = [\mathbf{a}_r (\phi_1 ), \mathbf{a}_r (\phi_2 ), \cdots,  \mathbf a_r (\phi_{G_R})], \ \mathbf A_T (\Phi_T ) = [\mathbf{a}_t (\phi_1 ), \mathbf{a}_t (\phi_2 ), \cdots~,  \mathbf a_t (\phi_{G_T})]. \label{eq:rx_array_res_dict}
\end{align}
Owing to the choice of grid angles considered in \eqref{eq:dir_cos_tx} and \eqref{eq:dir_cos_rx}, the matrices $\mathbf A_R (\Phi_R )$ and $\mathbf A_T (\Phi_T )$ are semi-unitary, i.e., they satisfy
\begin{align}
\mathbf{A}_i(\Phi_i)\mathbf{A}_i^H(\Phi_i)=\displaystyle\frac{G_i}{N_i} \mathbf{I}_{N_i}, ~~~~ i \in\{R,T\}. \label{eq:sem _unit}
\end{align} 
Using the above quantities, an equivalent angular-domain beamspace representation \cite{srivastava2019quasi,he2020beamspace} of the \ac{THz} \ac{MIMO} channel  $\mathbf{H}$ (c.f. \eqref{eq:ch_model}) can be obtained as 
\begin{align}
\mathbf H \simeq \mathbf A_{R} (\Phi_{R}) \mathbf H_{b} \mathbf A^H_{T}(\Phi_{T} ), \label{eq:beamspacechannel_matrix} \end{align}
where $\mathbf H_{b} \in \mathbb C^{G_R \times G_T}$ signifies the beamspace domain channel matrix. Note that when the grid sizes $G_R$ and $G_T$ are large, i.e., the quantization of \ac{AoA}/ \ac{AoD} grids is fine enough, the above approximate relationship holds with equality. Due to high free-space loss, as well as reflection and molecular absorption losses in a \ac{THz} system, the number of multipath components is significantly lower \cite{yan2019dynamic, sarieddeen2020overview, lin2015adaptive}. Furthermore, the \ac{THz} MIMO channel  comprises only a few highly-directional beams, which results in an angularly-sparse multipath channel. Hence, only a few active \ac{AoA}/ \ac{AoD} pairs exist in the channel, which makes the beamspace channel matrix $\mathbf H_{b}$ sparse in nature.

Once again, upon exploiting the properties of the matrix Kronecker product in \eqref{eq:beamspacechannel_matrix}, one obtains 
\begin{align}
\mathbf{h} = \mathrm{vec}(\mathbf H) = \left[ \mathbf A^{*}_T (\Phi_T ) \boldsymbol \otimes \mathbf A_R(\Phi_R) \right] \mathbf{h}_{b}, \label{eq: beamspacechannel_vector}
\end{align}
where $\mathbf{h}_{b} = \mathrm{vec}(\mathbf H_{b})\in\mathbb{C}^{G_RG_T \times 1}$. Finally, the sparse \ac{CSI} estimation model of the \ac{THz} \ac{MIMO} system can be obtained via substitution of \eqref{eq: beamspacechannel_vector} into \eqref{eq: equivalent_system_model}, yielding
\begin{align}
\mathbf y = \widetilde{\mathbf{\Phi}} \mathbf{h}_{b} + \mathbf v, \label{eq: CS_eq_model}
\end{align}
where $\widetilde{\mathbf{\Phi}} = \mathbf{\Phi} \mathbf{\Psi} \in \mathbb{C}^{M_TM_R \times G_RG_T}$ represents the equivalent sensing matrix, whereas $\mathbf{\Psi}= \left[ \mathbf A^{*}_T (\Phi_T ) \boldsymbol \otimes \mathbf A_R(\Phi_R) \right] \in \mathbb{C}^{N_RN_T \times G_RG_T}$ represents the \textit{sparsifying-dictionary}. Alternatively, one can express the equivalent sensing matrix $\widetilde{\mathbf{\Phi}}$ as
\begin{align}
    \widetilde{\mathbf \Phi} = \left[ \left(\mathbf{X}_{p}^T \mathbf{F}_{\text{RF}}^T \mathbf A^{*}_T (\Phi_T)\right)  \boldsymbol \otimes \left( \mathbf W^H_{\text{BB}} \mathbf W^H_{\text{RF}} \mathbf A_R(\Phi_R)\right) \right]. \label{eq:eq_sensing_mat}
\end{align}
It can be readily observed that the equivalent sensing matrix $\widetilde{\mathbf \Phi}$ depends on the choice of the RF TPC $\mathbf{F}_{\text{RF}}$, of the RF RC $\mathbf{W}_{\text{RF}}$, of the baseband RC $\mathbf{W}_{\text{BB}}$ and of the pilot matrix $\mathbf{X}_p$ employed for estimating the channel. Therefore, minimizing the total coherence
\cite{elad2007optimized, li2013projection} of the matrix $\widetilde{\mathbf \Phi}$ can lead to significantly improved sparse signal estimation. We now derive the optimal pilot matrix $\mathbf{X}_{p}$ and the baseband RC $\mathbf{W}_{\text{BB}}$, which achieve this.
\begin{lem}\label{lem:lemma1}
Let us set the RF TPC and RC to the normalized discrete Fourier transform (DFT) matrices as follows: $\mathbf{F}_{\text{RF}}\mathbf{F}_{\text{RF}}^H=\mathbf{F}_{\text{RF}}^H\mathbf{F}_{\text{RF}}=\mathbf{I}_{N_T}$ and $\mathbf{W}_{\text{RF}}\mathbf{W}_{\text{RF}}^H= \mathbf{W}_{\text{RF}}^H\mathbf{W}_{\text{RF}}=\mathbf{I}_{N_R}$. Then the $i$th diagonal block $\mathbf{X}_{p,i}, 1 \leq i \leq N_F$, of the pilot matrix $\mathbf{X}_{p}$ defined in \eqref{eq:X_p}, and $j$th diagonal block $\mathbf{W}_{\text{BB},j}, 1 \leq j \leq N_C$, of the baseband RC $\mathbf{W}_{\text{BB}}$ defined in \eqref{eq:BB_comb}, may be formulated as 
\begin{align}
\mathbf{X}_{p,i} = \mathbf{U} \Big[\mathbf{I}_{{\frac{M_T}{N_F}}} \ \ \mathbf{0}_{{{\frac{M_T}{N_F}}} \times N_{RF}-{{\frac{M_T}{N_F}}}}\Big]^T \mathbf{V}_1^H \text{ and }
\mathbf{W}_{\text{BB},j} = \mathbf{U} \Big[\mathbf{I}_{\frac{M_R}{N_C}}  \ \ \mathbf{0}_{{\frac{M_R}{N_C}} \times N_{RF}-{\frac{M_R}{N_C}}}\Big]^T \mathbf{V}_2^H,
\end{align}
for which the total coherence $\mu^t\left(\widetilde{\mathbf \Phi}\right)$ of the equivalent dictionary matrix $\widetilde{\mathbf \Phi}$ is minimized, where the matrices $\mathbf{U}, \mathbf{V}_1 \text{ and } \mathbf{V}_2$ are arbitrary unitary matrices of size $N_{RF} \times N_{RF}, {\frac{M_T}{N_F}} \times {\frac{M_T}{N_F}}$ and ${\frac{M_R}{N_C}} \times {\frac{M_R}{N_C}}$, respectively.
\label{lem:sensing_mat_lemma}
\end{lem}
\begin{proof}
Given in Appendix \ref{appen:proof_lem_1}.
\end{proof}
The next subsection develops an \ac{OMP}-based procedure for acquiring a sparse estimate of the \ac{THz} \ac{MIMO} channel exploiting the model of \eqref{eq: CS_eq_model}. 
\begin{algorithm}[t!]
\caption{\ac{OMP}-based sparse channel estimation for \ac{THz} \ac{MIMO} systems}\label{algo1}
\textbf{Input:} Equivalent sensing matrix $\widetilde{\mathbf{\Phi}}$, pilot output $\mathbf{y}$, array response dictionary matrices $\mathbf A_{R} (\Phi_{R})$ and $\mathbf A_{T}(\Phi_{T} )$, stopping parameter $\epsilon_t$\\
\textbf{Initialization:} Index set $\mathcal{I} = [\ ]$, $\widetilde{\mathbf{\Phi}}^{\mathcal{I}} = [\ ]$, residue vectors $\mathbf{r}_{-1} = \mathbf{0}_{M_TM_R \times 1},$ $\mathbf{r}_{0} = \mathbf{y} $, $\widehat{\mathbf{h}}_{b,\text{OMP}} = \mathbf{0}_{G_RG_T \times 1}$, counter $i = 0$\\
\textbf{while} $\left( \left\vert \parallel \mathbf{r}_{i-1} \parallel_2^2 - \parallel \mathbf{r}_{i} \parallel_2^2 \right\vert \ \geq \ \epsilon_t \right)$ \textbf{do}
\begin{enumerate}
\item $i \leftarrow i+1$
\item $j = \underset{k = 1,\cdots,G_TG_R}{\mathrm{arg\  max}} \big\vert \mathbf{r}_{i-1}^H\widetilde{\mathbf{\Phi}}(:,k)\big\vert$
\item $\mathcal{I} = \mathcal{I} \cup {j}$
\item $\widetilde{\mathbf{\Phi}}^{\mathcal{I}} = \widetilde{\mathbf{\Phi}}(:,\mathcal{I})$
\item $ \widehat{\mathbf{h}}^{i}_{\text {LS}}  = \left( \widetilde{\mathbf{\Phi}}^{\mathcal{I}} \right)^{\dagger} \mathbf{y}$
\item $\mathbf{r}_{i} = \mathbf{y} - \widetilde{\mathbf{\Phi}}^{\mathcal{I}}  \widehat{\mathbf{h}}^{i}_{\text {LS}} $
\end{enumerate}
\textbf{end while}\\
$\widehat{\mathbf{h}}_{b,\text{OMP}} \left( \mathcal{I} \right)  =  \widehat{\mathbf{h}}^{i}_{\text {LS}} $\\
\textbf{Output:} $\widehat{\mathbf{H}}_{\text{OMP}} = \mathbf A_{R} (\Phi_{R}) \mathrm{vec}^{-1}\left( \widehat{\mathbf{h}}_{b,\text{OMP}} \right) \mathbf A^H_{T}(\Phi_{T} )$
\end{algorithm}
\subsection{OMP-Based Sparse Channel Estimation in \ac{THz} Hybrid \ac{MIMO} Systems} \label{Sparse_CES}
The \ac{OMP}-based sparse channel estimation technique is summarized in Algorithm \ref{algo1} and its key steps are described next. Step-2 of each iteration identifies the specific column of the sensing matrix $\widetilde{\mathbf{\Phi}}$ that is maximally correlated with the residue $\mathbf{r}_{i-1}$ obtained in iteration $i-1$. Step-3 updates the index-set $\mathcal{I}$ by including the index $j$ obtained in Step-2. Subsequently, Step-4 determines the submatrix $\widetilde{\mathbf{\Phi}}^{\mathcal{I}}$ of the sensing matrix $\widetilde{\mathbf{\Phi}}$ using the specific columns, which are indexed by the set $\mathcal{I}$. Step-5 obtains the intermediate \ac{LS} solution $\widehat{\mathbf{h}}_{\text{LS}}^i$, while the associated residue vector $\mathbf{r}_i$ is computed using  $\widetilde{\mathbf{\Phi}}^{\mathcal{I}}$ in Step-6. These steps are repeated iteratively until the difference between the subsequent residuals becomes sufficiently small, i.e., $\left\vert \parallel \mathbf{r}_{i-1} \parallel_2^2 - \parallel \mathbf{r}_{i} \parallel_2^2 \right\vert \ < \ \epsilon_t$, where $\epsilon_t$ is a suitably chosen threshold. Finally, the \ac{OMP}-based estimate $\widehat{\mathbf{H}}_{\text{OMP}}$ of the \ac{THz} \ac{MIMO} channel using its beamspace estimate $\widehat{\mathbf{h}}_{b,\text{OMP}}$ is determined as $\widehat{\mathbf{H}}_{\text{OMP}} = \mathbf A_{R} (\Phi_{R}) \mathrm{vec}^{-1}\left( \widehat{\mathbf{h}}_{b,\text{OMP}} \right) \mathbf A^H_{T}(\Phi_{T} ).$ The key benefits of the proposed \ac{OMP} algorithm are its sparsity inducing nature and low computational cost. However, note that the choice of the stopping parameter $\epsilon_t$ plays a vital role in defining the convergence behavior of the \ac{OMP} algorithm, which renders the performance uncertain. Furthermore, this technique is susceptible to error propagation due to its greedy nature, since an errant selection of the index in a particular iteration cannot be corrected in the subsequent iterations. In view of these shortcomings of the OMP-based approach, the next subsection develops an efficient \ac{BL}-based framework, which leads to a significantly improved estimation accuracy of the \ac{THz} \ac{MIMO} channel.
\section{\ac{BL}-based Sparse Channel Estimation in \ac{THz} \ac{MIMO} Systems}\label{sec:BL_based_est} 
The proposed \ac{BL}-based sparse channel estimation technique relies on the Bayesian philosophy, which is especially well-suited for an under-determined system, where $M_TM_R << G_TG_R$. A brief outline of this procedure is as follows. The \ac{BL} procedure commences by assigning a parameterized Gaussian prior $f(\mathbf h_{b}; \mathbf{\boldsymbol\Gamma})$ to the sparse beamspace \ac{CSI} vector ${\mathbf{h}}_{b}$. The associated hyperparameter matrix $\mathbf{\boldsymbol\Gamma}$ is subsequently estimated by maximizing the Bayesian evidence $f(\mathbf{y}; \mathbf \Gamma)$. Finally, the \ac{MMSE} estimate of the beamspace channel is obtained using the estimated hyperparameter matrix $\widehat{\mathbf{\boldsymbol\Gamma}}$, which leads to an improved sparse channel estimate. The various steps are described in detail below.

Consider the parameterized Gaussian prior assigned to ${\mathbf{h}}_{b}$ as shown below \cite{wipf2004sparse}
\begin{align}
f(\mathbf h_{b}; \mathbf \Gamma)=\prod_{i=1}^{G_RG_T}(\pi \gamma_i)^{-1} \exp\Bigg(-\displaystyle\frac{\vert \mathbf h_{b}(i)\vert^2}{\gamma_i}\Bigg), \label{eq: gaussian_prior1}
\end{align}
where the matrix $\mathbf{\Gamma} {=} \mathrm{diag}\left(\gamma_1, \gamma_2, \cdots, \gamma_{G_R G_T}\right) {\in} \mathbb{R}^{G_RG_T \times G_R G_T}$ comprises the hyperparameters $\gamma_i, \: 1 \leq i \leq G_R G_T$. Note that the \ac{MMSE} estimate ${\widehat{\mathbf{h}}}_{b}$ corresponding to the sparse estimation model in \eqref{eq: CS_eq_model} is given by \cite{kay1993fundamentals}
\begin{align}
{\widehat{\mathbf{h}}}_{b} = \left( \widetilde{\mathbf{\Phi}}^H \mathbf{R}_{v}^{-1} \widetilde{\mathbf{\Phi}} + \mathbf{\Gamma}^{-1}\right)^{-1} \widetilde{\mathbf{\Phi}}^H \mathbf{R}_{v}^{-1} {\mathbf{y}}, \label{eq: mmse_est}
\end{align}
which can be readily seen to depend on the hyperparameter matrix $\mathbf{\Gamma}$. Therefore, the estimation of $\mathbf{\Gamma}$ holds the key for eventually arriving at a reliable sparse estimate of the beamspace channel vector ${\mathbf{h}}_{b}$. In order to achieve this, consider the log-likelihood function $\log\left[f(\mathbf{y};\mathbf \Gamma)\right]$ of the hyperparameter matrix $\mathbf\Gamma$, which can be formulated as $\log\left[ f(\mathbf{y};\boldsymbol\Gamma)\right] = {c}_1 - \log \left[\det\left(\mathbf{R}_{y}\right)\right] -  \mathbf y^H \mathbf{R}_{y}^{-1} \mathbf y,$ where we have ${c}_1 = -M_TM_R\log(\pi)$ and the matrix $\mathbf{R}_{y} = \mathbf{R}_{v} +  \widetilde{\mathbf{\Phi}} {\mathbf{\Gamma}} \widetilde{\mathbf{\Phi}}^H \in \mathbb{C}^{M_T M_R \times M_T M_R}$ represents the covariance matrix of the pilot output $\mathbf y$. It follows from \cite{wipf2004sparse} that maximization of the log-likelihood $\log\left[ f(\mathbf{y};\boldsymbol\Gamma)\right]$ with respect to $\mathbf{\Gamma}$ is non-concave, which renders its direct maximization intractable. Therefore, in such cases, the \ac{EM} technique is eminently suited for iterative maximization of the log-likelihood function, with guaranteed convergence to a local optimum \cite{kay1993fundamentals}. Let $\widehat{\gamma}_i^{(j-1)}$ denote the estimate of the $i$th hyperparameter obtained from the \ac{EM} iteration $(j-1)$ and let $\widehat{\mathbf{\Gamma}}^{(j-1)}$ denote the hyperparameter matrix defined as $\widehat{\mathbf{\Gamma}}^{(j-1)} = \mathrm{diag}\left( \widehat{\gamma}_1^{(j-1)}, \widehat{\gamma}_2^{(j-1)}, \cdots, \widehat{\gamma}_{G_RG_T}^{(j-1)} \right)$. The procedure of updating the estimate $\widehat{\mathbf{\Gamma}}^{(j)}$  in the $j$th \ac{EM}-iteration is described in Lemma \ref{lem:lemma2} below. 
\begin{lem} \label{lem:lemma2}
Given the $i$th hyperparameter $\widehat{\gamma}_i^{(j-1)}$, the update $\widehat{\gamma}_i^{(j)}$ in the $j$th \ac{EM}-iteration, which maximizes the log-likelihood function  $\mathbf{\mathcal{L}}\left(\mathbf{\Gamma}| \widehat{\mathbf{\Gamma}}^{(j-1)}\right)=\mathbb{E}_{{{\mathbf{h}}_{b}\vert{\mathbf{y}}; \widehat{\mathbf{\Gamma}}^{(j-1)}}} \left\{\log f({ \mathbf{y}}, {\mathbf{h}}_{b}; {\mathbf{\Gamma}} )\right\},$ is given by
\begin{align}
\widehat{\gamma}_i^{(j)} = \mathbf{R}_{b}^{(j)}(i,i) + \left\vert \boldsymbol{\mu}_{b}^{(j)}(i) \right\vert^2, \label{eq:hyp_comp}
\end{align}
where $\boldsymbol{\mu}_{b}^{(j)}= \mathbf{R}_{b}^{(j)} \widetilde{\mathbf{\Phi}}^H \mathbf{R}_{v}^{-1} {\mathbf{y}} \in \mathbb{C}^{G_R G_T \times 1}$ and  $\mathbf{R}_{b}^{(j)} = \left[ \widetilde{\mathbf{\Phi}}^H \mathbf{R}_{v}^{-1} \widetilde{\mathbf{\Phi}} + \left(\widehat{\mathbf{\Gamma}}^{(j-1)}\right)^{-1} \right]^{-1} \in \mathbb{C}^{G_R G_T \times G_R G_T}.$
\end{lem}
\begin{proof}
Given in Appendix \ref{appen:proof_lem_2}.
\end{proof}
The \ac{BL} algorithm of \ac{THz} \ac{MIMO} \ac{CSI} estimation is summarized in Algorithm \ref{algo2}. The \ac{EM} procedure is repeated until the estimates of the  hyperparameters converge, i.e. the quantity $\left\Vert \widehat{\boldsymbol{\Gamma}}^{(j)}-\widehat{\boldsymbol{\Gamma}}^{(j-1)} \right\Vert^{2}_{F}$ becomes smaller than a suitably chosen threshold $\epsilon$  or the number of iterations reaches a maximum limit $K_{\text{max}}$. The \ac{BL}-based estimate $\widehat{\mathbf h}_{b,\text{BL}}$, upon convergence of the \ac{EM} procedure, is given by $\widehat{\mathbf h}_{b,\text{BL}} = \boldsymbol{\mu}_{b}^{(j)}.$
\begin{algorithm}[t!]
\caption{\ac{BL}-based sparse channel estimation for \ac{THz} \ac{MIMO} systems}\label{algo2}
\textbf{Input:} Pilot output $\mathbf{y} $, equivalent sensing matrix $\widetilde{\mathbf{\Phi}}$, noise covariance $\mathbf{R}_{v}$, array response dictionary matrices $\mathbf A_{R} (\Phi_{R})$ and $\mathbf A_{T}(\Phi_{T} )$, stopping parameters $\epsilon \text{\ and\ } K_{\text{max}}$\\
\textbf{Initialization:} $\widehat{{\gamma}}_i^{(0)} =1$, for $1\leq i\leq G_RG_T \Rightarrow \widehat{\mathbf{\Gamma}}^{(0)} = \mathbf{I}_{G_RG_T}$, $\widehat{\mathbf{\Gamma}}^{(-1)} = \mathbf{0}$ and
 counter $j = 0$ \\
\textbf{while} $\left(\left\Vert \widehat{\boldsymbol\Gamma}^{(j)} - \widehat{\boldsymbol\Gamma}^{(j-1)}\right\Vert_F > \epsilon \text{\ \ and\ \ } j < K_{\text{max}}\right)$ \textbf{do}
\begin{itemize}
\item[1)] $j \leftarrow j+1$;
\item[2)] \textbf{E-step:} Compute the \textit{a posteriori} covariance and mean  
\begin{align*}
\mathbf{R}_{b}^{(j)} = \left[ \widetilde{\mathbf{\Phi}}^H \mathbf{R}_{v}^{-1} \widetilde{\mathbf{\Phi}} + \left(\widehat{\mathbf{\Gamma}}^{(j-1)}\right)^{-1} \right]^{-1} \text{ and }
\boldsymbol{\mu}_{b}^{(j)} = \mathbf{R}_{b}^{(j)} \widetilde{\mathbf{\Phi}}^H \mathbf{R}_{v}^{-1} {\mathbf{y}};
\end{align*}
\item[3)] \textbf{M-step:} Update the hyperparameters
\begin{itemize}
\item[] \textbf{for} $i=1,\cdots,G_RG_T$ \textbf{do}
\begin{align*}
\widehat{\gamma}_i^{(j)} = \mathbf{R}_{b}^{(j)}(i,i) + \left\vert \boldsymbol{\mu}_{b}^{(j)}(i) \right\vert^2
\end{align*}
\textbf{end for}
\end{itemize}
\end{itemize}
\textbf{end while}\\
\label{algo:M-SIP_algo}
 $\widehat{\mathbf h}_{b,\text{BL}}=\boldsymbol{\mu}_{b}^{(j)}$\\
\textbf{Output:}
$\widehat{\mathbf{H}}_{\text{BL}} = \mathbf A_{R} (\Phi_{R}) \mathrm{vec}^{-1}\left( \widehat{\mathbf{h}}_{b,\text{BL}} \right) \mathbf A^H_{T}(\Phi_{T} ).$
\end{algorithm}
\subsection{Multiple Measurement Vector (MMV)-BL (MBL)} 
Let us now consider a scenario with multiple pilot outputs, denoted by $\mathbf{y}_{m}, 1 \leq m \leq M$, which are obtained by transmitting an identical pilot matrix $\mathbf{X}_{p}$. The output $\mathbf{y}_{m}$ received from the $m$th pilot-block transmission can be formulated as $\mathbf{y}_{m} = \widetilde{\mathbf{\Phi}} \mathbf{h}_{b,m} + \mathbf{v}_{m}$, where $\mathbf{h}_{b,m}$ represents the beamspace channel corresponding to the $m$th pilot-block transmission and $\mathbf{v}_{m}$ is the corresponding noise. Now defining the concatenated matrices $\mathbf{\bar{Y}} = \left[ \mathbf{y}_{1} \cdots, \mathbf{y}_{M} \right] \in \mathbb{C}^{M_TM_R \times M}, \mathbf{\bar{H}}_{b}=\left[\mathbf{h}_{b,1}, \cdots, \mathbf{h}_{b,M} \right] \in \mathbb{C}^{G_RG_T \times M}$ and $\mathbf{\bar{V}}=\left[\mathbf{v}_1, \cdots, \mathbf{v}_M \right]$, the MMV model can be formulated as
\begin{align*}
\mathbf{\bar{Y}} = \widetilde{\mathbf{\Phi}} \mathbf{\bar{H}}_{b} + \mathbf{\bar{V}}. \label{eq:siso_mmv_2}
\end{align*}
Furthermore, considering these $M$ pilot-blocks to be well within the \textit{coherence-time} \cite{tse2005fundamentals}, the non-zero locations of the sparse beamspace domain CSI $\mathbf{h}_{b,m}$ do not change. This results in an interesting \textit{simultaneous-sparse} structure of the resultant CSI matrix $\mathbf{\bar{H}}_{b}$, since its columns $\mathbf{h}_{b,m}$ share an identical sparsity profile. Subsequently, one can derive an MBL framework for efficiently exploiting this simultaneous-sparse structure of the beamspace domain CSI matrix $\mathbf{\bar{H}}_b$, which yields a superior estimates. The update equations of the proposed MBL framework for the $j$th EM iteration is  summarized below \cite{wipf2007empirical}:
\begin{align*}
\mathbf{R}_{b}^{(j)} = \left[ \widetilde{\mathbf{\Phi}}^H \mathbf{R}_{v}^{-1} \widetilde{\mathbf{\Phi}} + \left(\widehat{\mathbf{\Gamma}}^{(j-1)}\right)^{-1} \right]^{-1} &\text{ and } \widehat{\mathbf{\bar{H}}}^{(j)} = \mathbf{R}_{b}^{(j)} \widetilde{\mathbf{\Phi}}^H \mathbf{R}_{v}^{-1} \mathbf{\bar{Y}},\\
\widehat{\gamma}_i^{(j)} = \mathbf{R}_{b}^{(j)}(i,i) + \frac{1}{M} & \sum_{m=1}^M \left\vert \widehat{\mathbf{\bar{H}}}^{(j)}(i,m) \right\vert^2.
\end{align*}
Finally, to benchmark the performance, the \ac{BCRLB} of the \ac{CSI} estimation model of \eqref{eq: CS_eq_model} is derived in the next subsection.
\subsection{\ac{BCRLB} for \ac{THz} \ac{MIMO} Channel Estimation}\label{sec:BCRLB}
The Bayesian \ac{FIM} $\mathbf{J}_{\text{B}} \in \mathbb{C}^{G_RG_T \times G_RG_T}$ can be evaluated as the sum of \acp{FIM} associated with the pilot output  $\mathbf{y}$ and the beamspace \ac{CSI} $\mathbf{h}_{b}$, denoted by $\mathbf{J}_{\text{D}}$ and $\mathbf{J}_{\text{P}}$, respectively. Hence, one can  express the Bayesian \ac{FIM} $\mathbf{J}_{\text{B}}$ as \cite{van2007bayesian}: $\mathbf{J}_{\text{B}}=\mathbf{J}_{\text{D}} + \mathbf{J}_{\text{P}},$ where the matrices $\mathbf{J}_{\text{D}}$ and $\mathbf{J}_{\text{P}}$ are determined as follows. Let the log-likelihoods corresponding to the \ac{THz} \ac{MIMO} beamspace channel $\mathbf{h}_{b}$ and the pilot output vector $\mathbf{y}$ be represented by $\mathcal{L}(\mathbf{h}_{b} ; \boldsymbol \Gamma)$ and $\mathcal{L}(\mathbf{y} {\mid} \mathbf{h}_{b})$, respectively. {These log-likelihoods simplify to  
\begin{align}
\mathcal{L}(\mathbf{y} {\mid} \mathbf{h}_{b}) &= \log \left[ f \left(\mathbf{y} {\mid} \mathbf{h}_{b}\right)\right]=c_2-\left(\mathbf{y}-\mathbf{\widetilde{\Phi}}{\mathbf{h}}_{b}\right)^H\mathbf{R}_{v}^{-1}  \left(\mathbf{y}-\mathbf{\widetilde{\Phi}}{\mathbf{h}}_{b}\right)\\
&= c_2 - \mathbf{y}^H \mathbf{R}_{v}^{-1} \mathbf{y} + \mathbf{h}_{b}^H \mathbf{\widetilde{\Phi}}^H \mathbf{R}_{v}^{-1} \mathbf{y} + \mathbf{y}^H \mathbf{R}_{v}^{-1} \mathbf{\widetilde{\Phi}}{\mathbf{h}}_{b} - \mathbf{h}_{b}^H \mathbf{\widetilde{\Phi}}^H \mathbf{R}_{v}^{-1} \mathbf{\widetilde{\Phi}} {\mathbf{h}}_{b}, \label{eq:log_likeli_simplified}\\
\mathcal{L}(\mathbf{h}_{b} ; \boldsymbol \Gamma) &= \log \left[ f(\mathbf{h}_{b} ; \boldsymbol \Gamma)\right] =  c_3 - {\mathbf{h}}_{b}^H {\boldsymbol\Gamma^{-1}} {\mathbf{h}}_{b},
\end{align}
where the terms $c_2 = -M_TM_R\log\left(\pi\right) -\log\left[\det(\mathbf{R}_{v}) \right]$ and $c_3 = -G_RG_T\log\left(\pi\right)-\log\left[\det\left(\boldsymbol \Gamma\right)\right]$ are constants that do not depend on the beamspace channel $\mathbf{h}_b$. The \acp{FIM} $\mathbf{J}_{\text{D}}$ and $\mathbf{J}_{\text{P}}$ expressed in terms of these log-likelihoods are defined as \cite{van2007bayesian}
\begin{align}
\mathbf{J}_{\text{D}} = {-\mathbb{E}_{\mathbf{y}, \mathbf{h}_{b}} \Bigg\{ \frac{\partial^2 \mathcal{L}(\mathbf{y} {\mid} \mathbf{h}_{b})}{\partial{\mathbf{h}}_{b}\partial{\mathbf{h}}_{b}^H} \Bigg\}},  \ \ \    
\mathbf{J}_{\text{P}} = {-\mathbb{E}_{{\mathbf{h}}_{b}} \Bigg\{ \frac{\partial^2  \mathcal{L}(\mathbf{h}_{b} ; \boldsymbol \Gamma) }{\partial{\mathbf{h}}_{b}\partial{\mathbf{h}}_{b}^H} \Bigg\}}.
\end{align}
The quantity $\mathbf{J}_{\text{D}}$ simplifies to $\mathbf{J}_{\text{D}} = \widetilde{\mathbf{\Phi}}^H \mathbf{R}_{v}^{-1} \widetilde{\mathbf{\Phi}}$, since the $(i,j)$th element of the Hessian matrix $\frac{\partial^2 \mathcal{L}(\mathbf{y} {\mid} \mathbf{h}_{b})}{\partial{\mathbf{h}}_{b}\partial{\mathbf{h}}_{b}^H}$ evaluated as $\frac{\partial^2 \mathcal{L}(\mathbf{y} {\mid} \mathbf{h}_{b})}{\partial{\mathbf{h}}_{b,i}\partial{\mathbf{h}}_{b,j}}$ becomes zero for the initial four terms of \eqref{eq:log_likeli_simplified}. As for the last term, the Hessian matrix evaluates to $\widetilde{\mathbf{\Phi}}^H \mathbf{R}_{v}^{-1} \widetilde{\mathbf{\Phi}}$. Similarly, the FIM $\mathbf{J}_{\text{P}}$ evalulates to $ \mathbf{J}_{\text{P}} =\boldsymbol{\Gamma}^{-1}$.} Thus, the Bayesian \ac{FIM} $\mathbf{J}_{\text{B}}$ evaluates to $\mathbf{J}_{\text{B}} = \widetilde{\mathbf{\Phi}}^H \mathbf{R}_{v}^{-1} \widetilde{\mathbf{\Phi}} + \boldsymbol{\Gamma}^{-1}.$ Finally, the \ac{MSE} of the estimate $\widehat{\mathbf{h}}_{b}$ can be bounded as 
\begin{align}
\mathrm{MSE}\left(\widehat{\mathbf{h}}_{b}\right) = \mathbb{E} \left\{ \left\Vert \widehat{\mathbf{h}}_{b} - \mathbf{h}_{b} \right\Vert^2 \right\} \geq \mathrm{Tr}\left\{\mathbf{J}_B^{-1}\right\} =\mathrm{Tr}\left\{ \left(\mathbf{\widetilde{\Phi}}^H \mathbf{R}_{v}^{-1} \widetilde{\mathbf{\Phi}} + \boldsymbol{\Gamma}^{-1}\right)^{-1} \right\}.
\end{align}
Furthermore, upon exploiting the relationship between the CSI vector $\mathbf{h}$ and its beamspace representation $\mathbf{h}_b$ given in \eqref{eq: beamspacechannel_vector}, one can express the \ac{BCRLB} for the estimated  \ac{CSI}  $\widehat{\mathbf{H}}$ as $\mathrm{MSE}\left(\widehat{\mathbf{H}}\right) \geq \mathrm{Tr}\left\{\mathbf{\Psi}\mathbf{J}_{\text{B}}^{-1}\mathbf{\Psi}^H\right\}.$
The next part of this paper presents the hybrid TPC and RC design using the CSI estimates obtained from the \ac{OMP} and \ac{BL} techniques described above.  
\section{Hybrid Transceiver Design for \ac{THz} \ac{MIMO} Systems} \label{prec_comb}
{This treatise develops a novel joint hybrid transceiver design, which directly employs the beamspace channel estimates obtained via the proposed BL-based estimation techniques. Note that the existing mmWave and THz contributions, such as \cite{yuan2018hybrid, alkhateeb2014channel, el2014spatially}, assume the availability of the full CSI for designing the RF precoder $\bar{\mathbf{F}}_{\text{RF}}$ and combiner $\bar{\mathbf{W}}_{\text{RF}}$, which is challenging to obtain due to the large number of antennas and propagation losses. Furthermore, these works either consider the true array response vectors to be perfectly known or employ a codebook for designing the RF precoder/ combiner. To the best of our knowledge, none of the existing papers have directly employed the estimate $\widehat{\mathbf{h}}_{b}$ of the underlying beamspace channel for designing the hybrid precoder, which is naturally the most suitable approach, given the availability of the  beamspace domain channel estimates. The proposed THz hybrid transceiver design addresses this open problem.}
\subsection{Hybrid TPC Design}
The baseband symbol vector $\bar{\mathbf x}$ of \eqref{eq:th_sys_mdl} comprised of i.i.d. symbols has a covariance matrix given by $\mathbf{R}_{\bar{x}}=\mathbb{E}\left\{\bar{\mathbf x} \bar{\mathbf x}^H\right\} = \frac{1}{N_s}\mathbf{I}_{N_S}.$ The transmit signal vector $\widetilde{\mathbf x} \in \mathbb{C}^{N_T \times 1}$ is formulated as $\widetilde{\mathbf x} = \bar{\mathbf F}_\text{RF} \bar{\mathbf{F}}_\text{BB} \bar{\mathbf x}$, while the power constraint on the hybrid TPC is given by $\| \bar{\mathbf F}_\text{RF} \bar{\mathbf F}_\text{BB}\|_F^2 \leq P_TN_S,$ which is equivalent to restricting the total transmit power at the output of the TPC to $P_T$, yielding $\mathbb{E}\left\{ \widetilde{\mathbf x}^H \widetilde{\mathbf x} \right\} = P_T$. To design the optimal TPCs $\bar{\mathbf{F}}_\text{BB}^{\mathrm{opt}} \in \mathbb{C}^{ N_\text{RF} \times N_S }\ \text{  and } \bar{\mathbf{F}}^{\mathrm{opt}}_\text{RF} \in \mathbb{C}^{N_T \times N_\text{RF}}$, one can maximize the mutual information $\mathcal{I}\left( \bar{\mathbf{F}}_\text{BB}, \bar{\mathbf{F}}_\text{RF} \right) = \mathrm{log}_2  \Big| \mathbf{I}_{N_R} + \mathbf{H} \bar{\mathbf{F}}_\text{RF} \bar{\mathbf{F}}_\text{BB} \bar{\mathbf{F}}_\text{BB}^H \bar{\mathbf{F}}_\text{RF}^H \mathbf{H}^H\Big|,$ subject to the power constraint. Thus, the optimization problem of the hybrid TPC can be formulated as
\begin{align}
\left\{\bar{\mathbf{F}}_\text{BB}^{\mathrm{opt}},\bar{\mathbf{F}}_\text{RF}^{\mathrm{opt}}\right\}
&= \underset{\left(\bar{\mathbf F}_\text{BB}, \bar{\mathbf F}_\text{RF}\right)}{{\mathrm{arg\text{ }  max}}} \ \mathrm{log}_2  \Big| \mathbf{I}_{N_R} + \mathbf{H} \bar{\mathbf{F}}_\text{RF} \bar{\mathbf{F}}_\text{BB} \bar{\mathbf{F}}_\text{BB}^H \bar{\mathbf{F}}_\text{RF}^H \mathbf{H}^H\Big|, \nonumber \\
& \hspace{-10pt} \mathrm{s. t.} \quad \|\bar{\mathbf{F}}_\text{RF} \bar{\mathbf{F}}_\text{BB}\|_F^2\leq P_TN_S, \vert \bar{\mathbf{F}}_{\text{RF}}(i,j)\vert = \frac{1}{\sqrt{N_T}}, 1 \leq i \leq N_T, 1 \leq j \leq N_{\text{RF}}. \label{eq:hybrid_precoder_optmzn}
\end{align}
Note that the above optimization problem is non-convex owing to the non-linear constraints on the elements of $\bar{\mathbf{F}}_\text{RF}$, which renders it intractable. To circumvent this problem, one can initially design the optimal fully-digital TPC $\bar{\mathbf{F}} \in \mathbb{C}^{N_T \times N_S}$ via the substitution $\bar{\mathbf{F}}_\text{RF} \bar{\mathbf{F}}_\text{BB} =\bar{\mathbf{F}}$ in the above optimization problem and ignoring the constant magnitude constraints. Upon obtaining the fully-digital TPC $\bar{\mathbf{F}}^{\mathrm{opt}}$, one can then decompose it into its RF and baseband constituents represented by the matrices $\bar{\mathbf{F}}_\text{RF}^{\mathrm{opt}}$ and $\bar{\mathbf{F}}_\text{BB}^{\mathrm{opt}}$, respectively. The well-known water-filling solution for design of the fully-digital TPC is as follows.

Let $\mathbf{H}=\mathbf{U}\boldsymbol{\Sigma}\mathbf{V}^H$ represent the \ac{SVD} of the \ac{THz} \ac{MIMO} channel. The optimal fully-digital TPC  $\bar{\mathbf{F}}^{\mathrm{opt}}$ is expressed as
\begin{align}\label{opt_prec}
\bar{\mathbf{F}}^{\mathrm{opt}}=\mathbf{V}_1\mathbf{P}^{{1}/{2}},
\end{align}
where $\mathbf{V}_1=\mathbf{V}\left(:,1:N_S\right)  \in \mathbb{C}^{N_T \times N_S}$ and the matrix $\mathbf{P} \in \mathbb{R}_{+}^{N_S \times N_S}$ represents a diagonal power allocation matrix, whose $i$th diagonal element $p_{i}, 1 \leq i \leq N_S$, can be derived as $p_{i} = \max\left\{0,\left( \lambda - \frac{\sigma_v^2}{(\boldsymbol\Sigma(i,i))^2}\right)\right\}.$ The quantity $\lambda$ denotes the Lagrangian multiplier \cite{boyd2004convex}, which satisfies the power constraint $ \sum_{i=1}^{N_S} p_{i} \leq P_T N_S$. Subsequently, the optimal hybrid TPCs $\bar{\mathbf{F}}_\text{BB}^{\mathrm{opt}}$ and $\bar{\mathbf{F}}_\text{RF}^{\mathrm{opt}}$ can be obtained from the  optimal fully-digital TPC $\bar{\mathbf{F}}^{\mathrm{opt}}$ as the solution of the approximate problem \cite{yan2019dynamic}
\begin{align} \label{opt_problem}
\left\{\bar{\mathbf{F}}_\text{BB}^{\mathrm{opt}}, \bar{\mathbf{F}}_\text{RF}^{\mathrm{opt}}\right\}&=\underset{\left(\bar{\mathbf{F}}_\text{BB}, \bar{\mathbf{F}}_\text{RF}\right)}{{\mathrm{arg\text{ }  min}}} \ \left\Vert \bar{\mathbf{F}}^{\mathrm{opt}}-\bar{\mathbf{F}}_\text{RF}\bar{\mathbf{F}}_\text{BB} \right\Vert_{F}^2, \nonumber \\
& \hspace{-10pt} \mathrm{s.t.} \quad \left\Vert \bar{\mathbf{F}}_\text{RF} \bar{\mathbf{F}}_\text{BB} \right\Vert_F^2 \leq P_TN_S, \ \ \left\vert \bar{\mathbf{F}}_{\text{RF}}(i,j) \right\vert = \frac{1}{\sqrt{N_T}}.
\end{align}
Although the above optimization problem is non-convex, the following interesting observation substantially simplifies hybrid TPC design. Note that the \ac{THz} \ac{MIMO} channel of \eqref{eq:ch_model} can be compactly represented as $\mathbf {H}= \bar{\mathbf{A}}_{R} \mathbf{D} \bar{\mathbf{A}}_{T}^H,$ where $\bar{\mathbf{A}}_{R} \in \mathbb{C}^{N_R \times \left(N_{\text{LoS}}+1\right)N_{\text{ray}}}$ and $\bar{\mathbf{A}}_{T} \in \mathbb{C}^{N_T \times \left(N_{\text{LoS}}+1\right)N_{\text{ray}}}$ are the matrices that comprise  $\left(N_{\text{LoS}}+1\right)N_{\text{ray}}$ array response vectors corresponding to the \ac{AoA}s and \ac{AoD}s of all the multipath components, respectively, whereas the diagonal matrix  $\mathbf{D} \in \mathbb{C}^{\left(N_{\text{LoS}}+1\right)N_{\text{ray}} \times \left(N_{\text{LoS}}+1\right)N_{\text{ray}}}$ contains their complex path-gains.
\begin{algorithm}[!t]
\textbf{Input:} Estimated beamspace channel  $\widehat{\mathbf{h}}_{b}$, optimal fully-digital TPC $\bar{\mathbf{F}}^\mathrm{opt}$ and MMSE RC $\bar{\mathbf{W}}_\text{M}$, output covariance matrix $\mathbf{R}_{yy}$, number of RF chains $N_\text{RF}$, array response dictionary matrices $\mathbf A_{R} (\Phi_{R})$ and $\mathbf A_{T}(\Phi_{T} )$ \\
\textbf{Initialization:} $\bar{\mathbf{F}}_\text{RF}=\left[ \ \right]$, $\bar{\mathbf{W}}_\text{RF}=\left[ \ \right]$,  $\mathbf{h}_{b,\mathrm{abs}}=\vert\widehat{\mathbf{h}}_b\vert$, construct an ordered set $\mathcal{S}$ from the indices of elements of the vector $\mathbf{h}_{b,\mathrm{abs}}$, so that  $\mathbf{h}_{b,\mathrm{abs}}\left[\mathcal{S}(1)\right]\geq \mathbf{h}_{b,\mathrm{abs}}\left[\mathcal{S}(2)\right] \geq \mathbf{h}_{b,\mathrm{abs}}\left[\mathcal{S}(3)\right]\geq, \cdots \geq\mathbf{h}_{b,\mathrm{abs}}\left[\mathcal{S}(G_RG_T)\right] $\\
\textbf{for}\ {$i=1,2,\hdots,N_\text{RF}$}
\begin{itemize}
\item[1)] $j=\mathrm{floor}\left[ \left(\mathcal{S}\left(i\right)-1\right)/G_R\right]+1; k=\mathrm{rem}\left[ \left(\mathcal{S}\left(i\right)-1\right),G_R\right]+1;$ 
\item[2)] $\bar{\mathbf{F}}_\text{RF}=\left[\bar{\mathbf{F}}_\text{RF} \quad \mathbf{a}_t\left(\phi_j\right)\right]; \bar{\mathbf{W}}_\text{RF}=\left[\bar{\mathbf{W}}_\text{RF} \quad \mathbf{a}_r\left(\phi_k\right)\right];$
\end{itemize}
\textbf{end for}\\
$\bar{\mathbf{F}}_\text{BB}={\left(\bar{\mathbf{F}}_\text{RF}\right)}^{\dagger}\bar{\mathbf{F}}^\mathrm{opt}; \bar{\mathbf{W}}_\text{BB}=\left( \bar{\mathbf{W}}_\text{RF}^H\mathbf{R}_{yy}\bar{\mathbf{W}}_\text{RF}\right)^{-1}\bar{\mathbf{W}}_\text{RF}^H\mathbf{R}_{yy}\bar{\mathbf{W}}_\text{M}$; \\
\textbf{Output:} $\bar{\mathbf{F}}_\text{BB}$, $\bar{\mathbf{F}}_\text{RF}$, $\bar{\mathbf{W}}_\text{BB}$, $\bar{\mathbf{W}}_\text{RF}$
\caption{Hybrid transceiver design from the estimated beamspace \ac{THz} \ac{MIMO} channel $\widehat{\mathbf{h}}_{b}$}
\label{algo3}
\end{algorithm}
Thus, the row- and column-spaces of the channel matrix $\mathbf{H}$ obey
\begin{align}
    \mathcal{R}\left(\mathbf{H}^*\right) = \mathcal{C} \left( \bar{\mathbf A}_{T} \right), \quad \mathcal{C}\left( \mathbf{H}\right) = \mathcal{C} \left( \bar{\mathbf A}_{R} \right), \label{eq:spaces_rel_1}
\end{align}
where $\mathcal{R}(\cdot)$ and $\mathcal{C}(\cdot)$ represent the row and column spaces, respectively, of a matrix. At this juncture, using the \ac{SVD} of $\mathbf{H}$ together with \eqref{opt_prec}, one can conclude that
\begin{align}
   \mathcal{C}\left( \mathbf{F}^\mathrm{opt}\right) \subseteq \mathcal{C} \left(\mathbf V_1 \right) \subseteq \mathcal{C} \left(\mathbf V(:,1:\rho) \right) = \mathcal{R}\left(\mathbf{H}^*\right), \label{eq:spaces_rel_2}
\end{align}
where we have $\rho=\mathrm{rank}(\mathbf{H})$ and $\rho \geq N_S$. Hence, from \eqref{eq:spaces_rel_1} and \eqref{eq:spaces_rel_2}, one can deduce that 
\begin{align}
    \mathcal{C}\left( \bar{\mathbf{F}}^\mathrm{opt}\right) \subseteq \mathcal{C} \left( \bar{\mathbf A}_{T} \right).
\end{align}
This implies that a suitable linear combination of the columns of $\bar{\mathbf A}_{T}$ can determine any column of the matrix $\bar{\mathbf{F}}^\mathrm{opt}$. Furthermore, since it is evident that the array response vectors $\mathbf{a}_t$ also satisfy the non-convex constraints of \eqref{opt_problem}, courtesy \eqref{eq:Array_resp_vector}, the columns of the matrix $\bar{\mathbf A}_{T}$ are a suitable candidate for the \ac{RF} TPC $\bar{\mathbf{F}}_{\text{RF}}$. However, a pair of key challenges remain. Firstly, the array response matrix $\bar{\mathbf A}_{T}$ is unknown. To compound this problem, one can only choose $N_\text{RF}$ columns of $\bar{\mathbf A}_{T}$ for the design of the \ac{RF} TPC, owing to the fact that there are only $N_\text{RF}$ \ac{RF} chains. Both the above-mentioned issues can be efficiently addressed by employing the estimate $\widehat{\mathbf{h}}_b$ of the beamspace channel as follows.

Note that the dominant coefficients of the beamspace channel matrix $\mathbf{H}_b$ (c.f. \eqref{eq:beamspacechannel_matrix}) represent the active (\ac{AoA}, \ac{AoD})-pairs. Therefore,  to design the \ac{RF} TPC $\bar{\mathbf{F}}_{\text{RF}}$, one can directly employ the estimate $\widehat{\mathbf{H}}_b$ of the beamspace channel matrix derived from the estimation techniques proposed in Section-\ref{Sparse_CES} and \ref{sec:BL_based_est}. The salient steps in the proposed hybrid transceiver design procedure are detailed in Algorithm \ref{algo3}. We commence by arranging the elements of the quantity $\vert \widehat{\mathbf{h}}_b \vert$ in descending order and determine the $N_\text{RF}$ entries that have the highest magnitude. The corresponding locations in the beamspace matrix representation yield the active (\ac{AoA}, \ac{AoD})-pairs. More precisely, the column indices, represented by $j$ in Step-1 of Algorithm \ref{algo3}, provide the active \acp{AoD} in the transmit angular grid $\Phi_T$, whereas the row indices, denoted by $k$, yield the active \acp{AoA}. The \ac{RF} TPC $\bar{\mathbf{F}}_{\text{RF}}$ can be subsequently constructed from the $N_\text{RF}$-dominant columns of the transmit array response dictionary matrix $\mathbf A_T (\Phi_T )$ (c.f. \eqref{eq:rx_array_res_dict}). Finally, the baseband TPC $\bar{\mathbf{F}}_\text{BB}$ can be obtained from the \ac{LS} estimate as $\bar{\mathbf{F}}_\text{BB}={\left(\bar{\mathbf{F}}_\text{RF}\right)}^{\dagger}\bar{\mathbf{F}}^\mathrm{opt}$. The procedure of the hybrid RC design in \ac{THz} \ac{MIMO} systems is described next.
\subsection{Hybrid \ac{MMSE} RC Design}
This subsection describes the design of the hybrid \ac{MMSE} RC components $\bar{\mathbf{W}}_\text{BB} \in \mathbb{C}^{ N_\text{RF} \times N_S }$ and $\bar{\mathbf{W}}_\text{RF} \in \mathbb{C}^{N_R \times N_\text{RF}}$ relying on the estimated beamspace channel matrix $\widehat{\mathbf{H}}_b$. Toward this, for a given hybrid TPC $\bar{\mathbf{F}}_\text{BB}$  and  $\bar{\mathbf{F}}_\text{RF}$, one can minimize the \ac{MSE} of approximation between the transmit baseband symbol vector $\bar{\mathbf x} \in \mathbb{C}^{N_S \times 1}$ and the output $\bar{\mathbf y}$, which obey \eqref{eq:th_sys_mdl}, subject to the constant-magnitude constraints on the elements of the RF RC $\bar{\mathbf{W}}_\text{RF}$. Let $\mathbf{y} \in \mathbb{C}^{N_R \times 1}$ denote the signal impinging at the RAs, which is given by $$\mathbf{y} =  \mathbf{H} \bar{\mathbf{F}}_{\text{RF}} \bar{\mathbf{F}}_{\text{BB}} \bar{\mathbf x} + \bar{\mathbf v}.$$ Thus, the RC design optimization problem can be formulated as
\begin{align}
\left\{\bar{\mathbf{W}}_\text{RF}^{\mathrm{opt}},\bar{\mathbf{W}}_\text{BB}^{\mathrm{opt}}\right\}
&= \underset{\left(\bar{\mathbf{W}}_\text{RF},\bar{\mathbf{W}}_\text{BB}\right)}{{\mathrm{arg\text{ }  min}}} \  \mathbb{E}  \bigg\{ \left\Vert\bar{\mathbf x}-\bar{\mathbf{W}}_\text{BB}^H \bar{\mathbf{W}}_\text{RF}^H\mathbf y\right\Vert^{2}_{2}\bigg\},   \nonumber \\
&  \hspace{25pt}\mathrm{s. t.} \quad \left\vert\bar{\mathbf{W}}_{\text{RF}}(i,j)\right\vert = \frac{1}{\sqrt{N_R}}. \label{eq:hybrid_combiner_optmzn}
\end{align}
As detailed in Appendix \ref{appen:proof_MMSE_comb}, the above optimization problem can be reformulated as
\begin{align}
\left\{\bar{\mathbf{W}}_\text{RF}^{\mathrm{opt}},\bar{\mathbf{W}}_\text{BB}^{\mathrm{opt}}\right\}
&= \underset{\left(\bar{\mathbf{W}}_\text{RF},\bar{\mathbf{W}}_\text{BB}\right)}{{\mathrm{arg\text{ }  min}}} \  \left\Vert\mathbf{R}_{yy}^{1/2}\left(\bar{\mathbf{W}}_\text{M}-\bar{\mathbf{W}}_\text{RF} \bar{\mathbf{W}}_\text{BB}\right)\right\Vert_{F}^2   \nonumber \\
&   \hspace{25pt}\mathrm{s. t.} \quad \left\vert\bar{\mathbf{W}}_{\text{RF}}(i,j)\right\vert = \frac{1}{\sqrt{N_R}},\label{opt_combiner_eq}
\end{align}
where the matrices $\mathbf{R}_{yy} \in \mathbb{C}^{ N_R \times N_R }$ and $\bar{\mathbf{W}}_\text{M} \in \mathbb{C}^{ N_R \times N_S}$ represent the covariance matrix of the output vector $\mathbf{y}$ and the optimal \ac{MMSE} RC, respectively. These can be formulated as
\begin{align}
\mathbf{R}_{yy}=&\mathbf{E}\left\{\mathbf{y}\mathbf{y}^H\right\}= \frac{1}{N_S}\left( \mathbf{H} \bar{\mathbf{F}}_\text{RF} \bar{\mathbf{F}}_\text{BB}\bar{\mathbf{F}}_\text{BB}^H\bar{\mathbf{F}}_\text{RF}^H\mathbf{H}^H+N_S\sigma_v^2 \mathbf{I}_{N_R} \right),\\
\bar{\mathbf{W}}_\text{M}=&\mathbf{H} \bar{\mathbf{F}}_\text{RF} \bar{\mathbf{F}}_\text{BB}\left(\bar{\mathbf{F}}_\text{BB}^H\bar{\mathbf{F}}_\text{RF}^H\mathbf{H}^H\mathbf{H}\bar{\mathbf{F}}_\text{RF}\bar{\mathbf{F}}_\text{BB}+N_S\sigma_v^2\mathbf{I}_{N_S}\right)^{-1}.\label{opt_mmse_cominer}
\end{align}
Since we have $\mathcal{C}\left( \bar{\mathbf{W}}_\text{M} \right) \subseteq \mathcal{C}\left( \mathbf{H} \right) =  \mathcal{C}\left( \bar{\mathbf{A}}_{R} \right)$, similar to the simplified TPC design, one can design the RF RC $\bar{\mathbf{W}}_\text{RF}$ from the array response vectors of the $N_{\text{RF}}$ active \acp{AoA} obtained from the estimated beamspace channel. Finally, the baseband RC $\bar{\mathbf{W}}_\text{BB}$ can be derived using the following weighted-\ac{LS} solution: $\bar{\mathbf{W}}_\text{BB} =  \left( \bar{\mathbf{W}}_\text{RF}^H\mathbf{R}_{yy}\bar{\mathbf{W}}_\text{RF}\right)^{-1}\bar{\mathbf{W}}_\text{RF}^H\mathbf{R}_{yy}\bar{\mathbf{W}}_\text{M}.$ For convenience, the hybrid RC design is also presented in Algorithm \ref{algo3}. Note that a key advantage of the proposed hybrid \ac{MMSE} RC design is that the processed signal $\bar{\mathbf y} = \bar{\mathbf{W}}_\text{BB}^H \bar{\mathbf{W}}_\text{RF}^H\mathbf y$ directly yields the MMSE estimate of the transmit symbol vector $\bar{\mathbf{x}}$.

{Note that the SOMP technique, as described in \cite{el2014spatially}, requires $N_{\text{RF}}$ iterations for selecting the $N_{\text{RF}}$ dominant array response vectors via a computationally intensive correlation method (Step-4 and 5 of Algorithm-1 in \cite{el2014spatially}), followed by an intermediate LS solution in each iteration. By contrast, the proposed hybrid precoder design framework is directly able to compute the final baseband precoder using the LS solution, once the RF precoder is derived using the estimated beamspace domain CSI. Thus, the proposed hybrid precoder design has a significantly lower computational cost, while performing very close to the ideal fully-digital benchmark, as demonstrated in our simulation results of Fig. \ref{fig:THz_ch_capacity_plots} and Fig. \ref{fig:THz_ch_capacity_BER_plots}. Furthermore, the framework for beamspace domain CSI estimation, followed by our hybrid transceiver design developed requires significantly lower feedback, since the receiver only has to feed back a few indices of the dominant beamspace components together with their quantized gains in order to construct the hybrid precoder of the transmitter.}

{The objectives of the proposed hybrid transceiver optimization problems in \eqref{eq:hybrid_precoder_optmzn} and \eqref{eq:hybrid_combiner_optmzn} of this treatise are to design a capacity-optimal hybrid precoder and MMSE-optimal hybrid combiner. We would like to clarify that the proposed solution does not guarantee optimality, since our solution directly employs the estimate $\widehat{\mathbf{h}}_b$ of the beamspace domain channel obtained from the proposed BL-based \ac{CSI} estimators. Hence, its performance heavily relies on the estimated CSI, as demonstrated in our simulation results of Fig. \ref{fig:THz_ch_capacity_plots} and \ref{fig:THz_ch_capacity_BER_plots}, which will always be the case for any practical solution developed for this problem. However, a solid mathematical foundation established after \eqref{opt_problem} justifies its low complexity and significantly improved performance that is close to the corresponding optimal fully-digital solution. On the other hand, the existing optimization algorithms conceived for hybrid transceiver design, such as \cite{8329410, 8606437}, may guarantee certain optimality, but they are typically iterative and computationally complex.}

{
\subsection{Computational Complexity} \label{subsec:comp_comp}
This subsection derives the computational cost of the proposed THz hybrid MIMO transceiver design, which is directly coupled with the beamspace domain CSI estimation module. The computational complexity order of the BL technique may be shown to be $\mathcal{O}\left( G_R^3G_T^3 \right)$, which arises due to the matrix inversion of size-$\left[ G_RG_T \times G_RG_T \right]$. On the other hand, the worst-case complexity order of the OMP scheme is seen to be $\mathcal{O}\left( M_T^3 M_R^3 \right)$, which arises due to the intermediate LS estimate required in each iteration. Finally, the computational cost of the hybrid transceiver design presented in Algorithm-\ref{algo3} based on the estimated BL-based CSI is seen to be on the order of $\mathcal{O}\left( N_T^3 + N_{\text{RF}}^3 +N_S^3 \right)$. Here, the $\mathcal{O}\left( N_T^3 \right)$ term arises due to the SVD of the THz MIMO channel $\mathbf{H}$, $\mathcal{O}\left( N_{\text{RF}}^3 \right)$ is due to the LS solution of the baseband precoder $\bar{\mathbf{F}}_\text{BB}$ and combiner $\bar{\mathbf{W}}_\text{BB}$, whereas $\mathcal{O}\left( N_S^3 \right)$ is due to the calculation of the fully-digital MMSE solution in \eqref{opt_mmse_cominer}. Thus, it can be readily observed that the overall computational cost of obtaining the OMP-based estimated CSI followed by the hybrid transceiver design is lower than that of employing the BL-based estimated CSI. However, as discussed later in our simulation results, the performance of the proposed hybrid transceiver design using OMP-based CSI is poor in comparison to that obtained via the BL-based CSI for an identical pilot overhead. Hence, there is a trade-off between the computational cost and the performance improvement attained.}
\section{Simulation Results}\label{sec:simulation_results}
The performance of the proposed \ac{CSI} estimation techniques conceived for our hybrid \ac{THz} \ac{MIMO} transceiver design is illustrated by our simulation results. For this study, the magnitudes of the \ac{LoS} and \ac{NLoS} complex path-gains $\alpha(f,d)$ have been generated using \eqref{eq:Los_pg} and \eqref{eq:Nlos_pg}, respectively, whereas the associated  phase shifts $\psi$ are generated as \ac{i.i.d.} samples of a random variable uniformly distributed over the interval $(-\pi, \pi]$. The molecular absorption coefficient $k_{\text{abs}}(f)$ has been computed using the procedure described in Section-\ref{sec:k_abs} relying on the \ac{HITRAN} database \cite{rothman2009hitran}. The operating carrier frequency $f$ and the transmission distance $d$ are set to $0.3$\,\ac{THz} and $10$\,m, respectively, unless stated otherwise. Furthermore, an office scenario is considered with the system pressure $p$ and temperature $T$ set to $1$\,atm and $296$\,K, respectively, which has the following molecular composition: water vapour $($1\%$)$, oxygen $($20.9\%$)$ and nitrogen $($78.1\%$)$. 
The \ac{THz} \ac{MIMO} channel is generated using a single \ac{LoS} and $N_\text{NLoS}=4$ \ac{NLoS} components, in which $3$ \ac{NLoS} components have first-order reflections, whereas the $4$th \ac{NLoS} component is assumed to have a second-order reflection from the respective scatterer. Furthermore, each multipath component is composed of $N_\text{ray} \in \{1,3\}$ diffused rays, whose \acp{AoA}/ \acp{AoD} follow \ac{i.i.d.} Gaussian distributions with an angular spread of standard deviation of $1/10\text{ radian}$ around the mean angle of the particular multipath component \cite{priebe2011aoa}. The standard deviation of the roughness of various reflecting media is set as $\sigma \in \{0.05,0.13,0.15\}$\,mm \cite{piesiewicz2007scattering}. The TA and RA gains, $G_t^{a}$ and $G_r^{a}$, respectively, are set to $G_t^a=G_r^a= 25$\,dB. Given the various channel parameters mentioned above, the \ac{THz} \ac{MIMO} channel has been generated using \eqref{eq:ch_model}-\eqref{eq:NLOS_model}.

For simulation, this work considers two \ac{THz} \ac{MIMO} systems, namely System-I and System-II, having the simulation parameters described below. For System-I, the number of TAs/ RAs is set to $N_T=N_R=32$ with $N_\text{RF}=8$ RF chains at both the ends. The number of training vectors, $M_T$ and $M_R$, is set to $M_T=M_R=24$, which can be seen to be lower than  $N_T$ and $N_R$. The angular grid sizes, $G_T$ and $G_R$, for this system are set as $G_T=G_R=36$, which is higher than $\max(N_T, N_R)$. By contrast, the simulation parameters of System-II are as follows: $N_T=N_R=16$, $N_\text{RF}=4$, $M_T=M_R=12$ and $G_T=G_R=20$. Note that, in contrast to the conventional channel estimation models, which are typically over-determined, the setting for System-I results in a $[576 \times 1296]$-size equivalent sensing matrix $\widetilde{\Phi}$, thus leading to an under-determined system, as described by Eq. \eqref{eq: CS_eq_model}. However, as shown in the simulation results, the proposed sparse estimation techniques developed in our paper are able to estimate the \ac{THz} \ac{MIMO} \ac{CSI} with the desired accuracy even in such a challenging scenario.  Furthermore, the antenna spacings, $d_t$ and $d_r$, for both the Systems have been set to $d_t=d_r= \frac{\lambda}{2}$. The SNR is defined as $\text{SNR}=10\log_{10}\left(\frac{1}{\sigma_v^2}\right)\text{ dB}$. For the \ac{OMP} technique, the stopping parameter $\epsilon_t$ is set to $\epsilon_t=\sigma_v^2$, whereas for the \ac{BL} technique, we set $\epsilon=10^{-6}$ and $K_\text{max}=50$.
\begin{figure*}[t!]
\centering
\subfloat[]{\includegraphics[scale = .43]{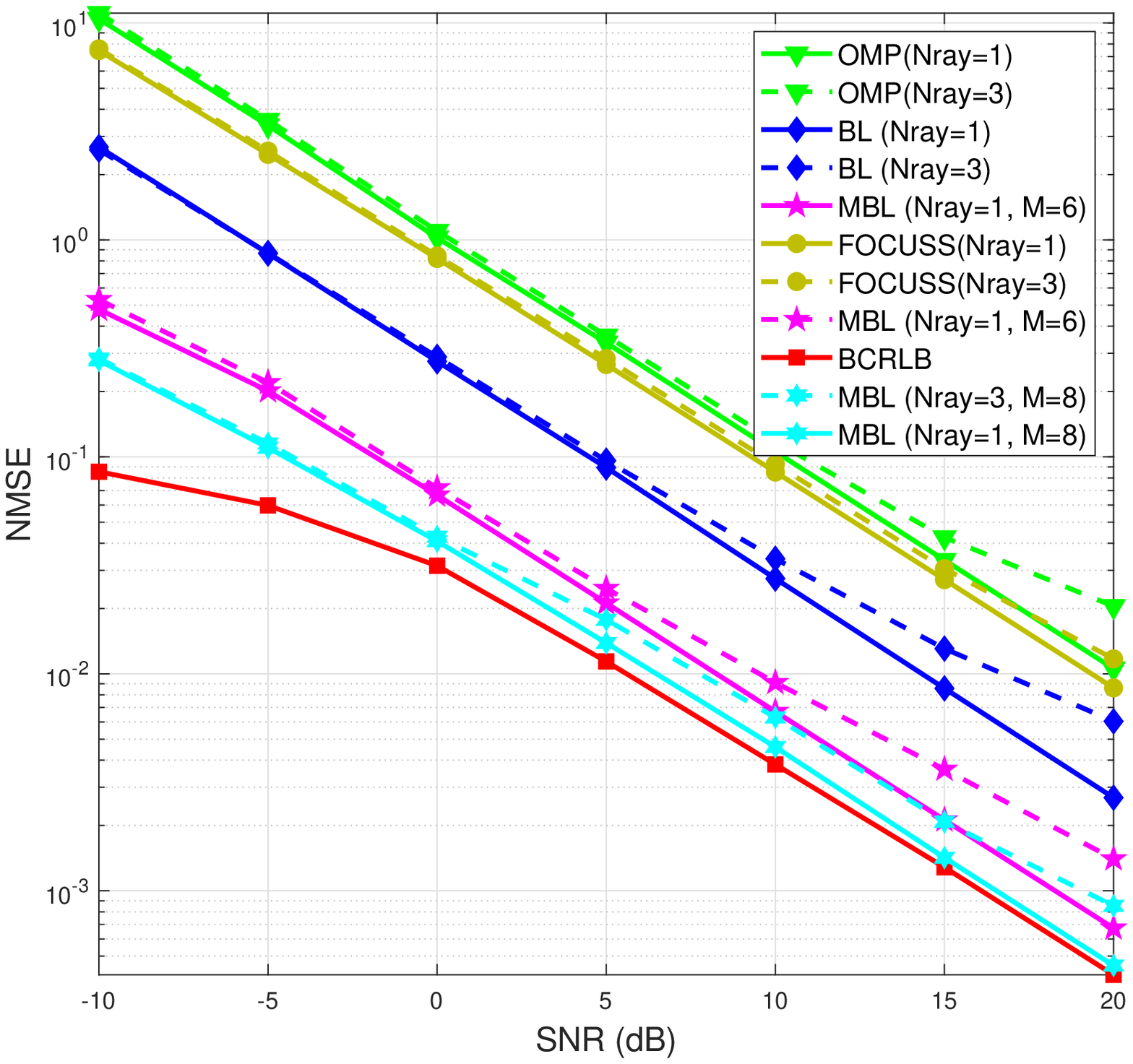}}
\hfil
\hspace{-20pt}\subfloat[]{\includegraphics[scale = .43]{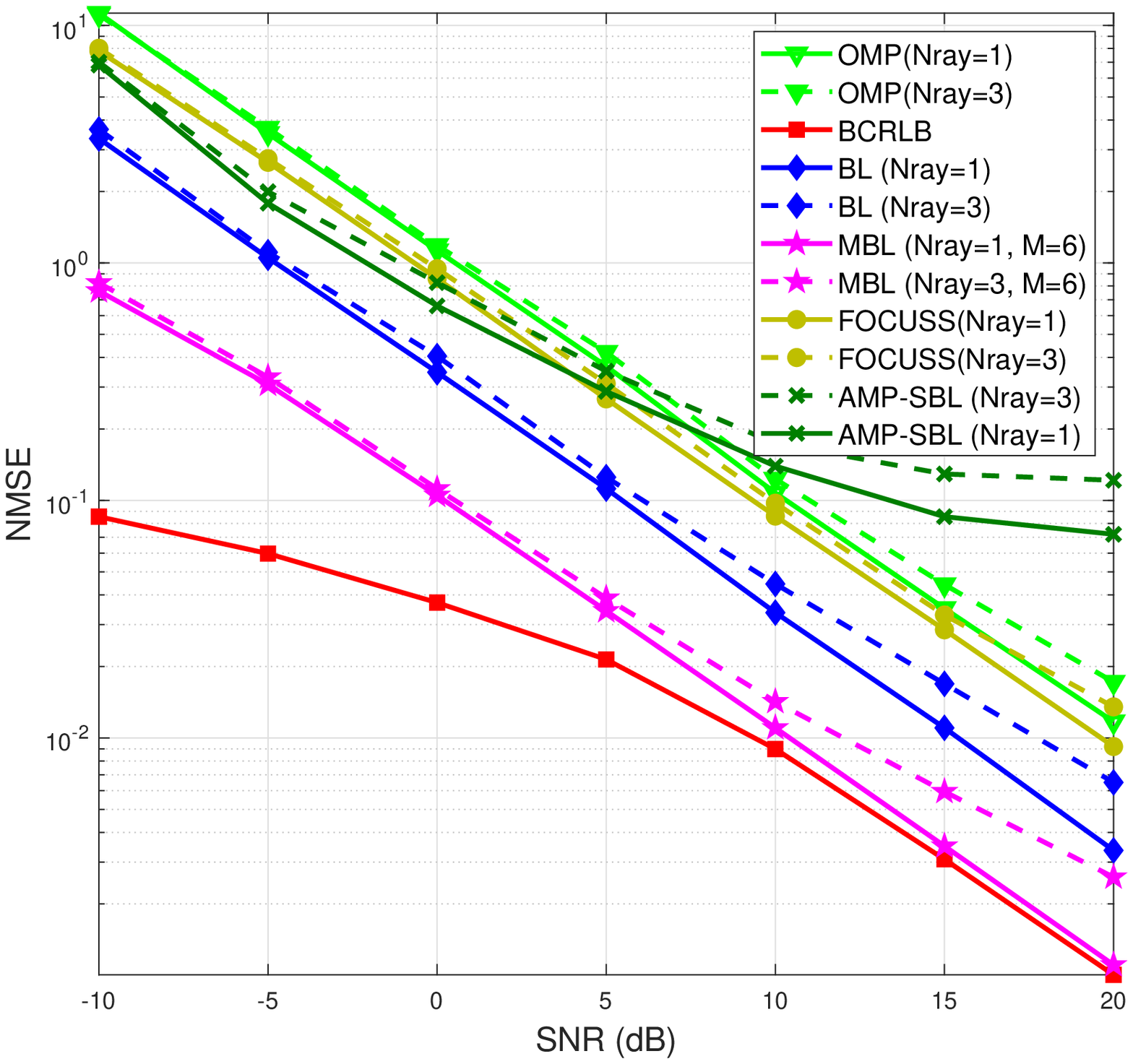}}
\caption{ NMSE versus SNR comparison for a \ac{THz} \ac{MIMO} $ \left(a\right)$  System-I; $\left(b\right)$ System-II.}
\label{fig:THz_ch_est_plots}
\end{figure*}
\subsection{\ac{THz} Hybrid \ac{MIMO} Channel Estimation}
Fig. \ref{fig:THz_ch_est_plots}(a) and Fig. \ref{fig:THz_ch_est_plots}(b) illustrate the sparse channel estimation performance versus SNR for the \ac{THz} \ac{MIMO} System-I and System-II, respectively, in terms of the \ac{NMSE}, which is defined as $\mathrm{NMSE}=\frac{\left\Vert\widehat{\mathbf{H}}-\mathbf{H}\right\Vert_F^2}{\left\Vert\mathbf{H}\right\Vert_F^2}.$ The performance of the proposed \ac{OMP} and \ac{BL}-based algorithms is also compared to that of the popular sparse signal recovery technique \ac{FOCUSS} \cite{gorodnitsky1997sparse}, typically used in the field of image reconstruction. The performance of all the competing techniques is also benchmarked against the \ac{BCRLB}, as derived in Section-\ref{sec:BCRLB}. 
From both the figures, one can conclude that the proposed \ac{BL}-based sparse channel estimation technique outperforms the \ac{OMP} and \ac{FOCUSS}, which is attributed to its robustness toward the tolerance parameter $\epsilon$ and $K_\text{max}$, and toward the dictionary matrix $\widetilde{\mathbf{\Phi}}$. On the other hand, the sensitivity of the \ac{OMP} technique to the stopping threshold $\epsilon_t$ and to the dictionary matrix lead to structural and convergence errors, as described in \cite{wipf2004sparse}, thus degrading the eventual sparse recovery of the beamspace channel. Furthermore, the \ac{OMP} technique suffers due to its greedy nature and error propagation, since the error encountered in the selection of the indices cannot be rectified in the subsequent iterations, thus negatively impacting its performance. On the other hand, the performance of \ac{FOCUSS} is poor due to its convergence deficiencies and sensitivity to the regularization parameter \cite{wipf2004sparse}. {The proposed techniques are also compared to low-complexity approximate message passing (MP)-based sparse Bayesian learning (AMP-SBL) \cite{al2017gamp}, which is the Bayesian extension of the MP algorithms developed in \cite{8523816, 9366805}. The performance of the AMP-SBL algorithm is poor in comparison to the proposed BL algorithm, since it only tracks the \textit{a posteriori} mean and variance of each element of the sparse vector, leading to its sub-optimal performance, especially at high SNR.} {One can also note from Fig. \ref{fig:THz_ch_est_plots}(a) that the proposed \ac{MBL} technique approaches the \ac{BCRLB} upon increasing the number of measurements $M$.} This is significant, since the \ac{BCRLB} is derived for an ideal scenario, where the \acp{AoA}/ \acp{AoD} are perfectly known, whereas the \ac{BL} framework does not rely on this idealized simplifying assumption. Another interesting observation is as follows. When the \ac{THz} \ac{MIMO} channel has $N_{ray} = 3$ diffused rays, the performance of all the competing schemes degrades. The reason behind this degradation is that the diffused rays lead to broadening the beamwidth of the \acp{AoA}/ \acp{AoD}, which essentially increases the support of the beamspace channel, eventually degrading the performance of sparse signal recovery. However, one can also verify that this degradation is minimal for the proposed \ac{BL} scheme, which outperforms the others in this scenario as well. Furthermore, one can also note that the proposed sparse estimation frameworks are capable of accurately estimating the $N_R \times N_T$ \ac{THz} \ac{MIMO} channel using $M_T$ and $M_R$ beam-patterns, where $M_TM_R << N_TN_R$. It is plausible that this is not possible using the conventional LS and \ac{MMSE} schemes, as described in Section-\ref{sec:channel_est}. Thus, its superior \ac{CSI} estimation performance coupled with its lower pilot overhead make the proposed \ac{BL}-based sparse estimation framework ideally suited for \ac{THz} \ac{MIMO} systems.

\begin{figure*}[t!]
\centering
\hspace{-10pt}\subfloat[]{\includegraphics[scale = .33]{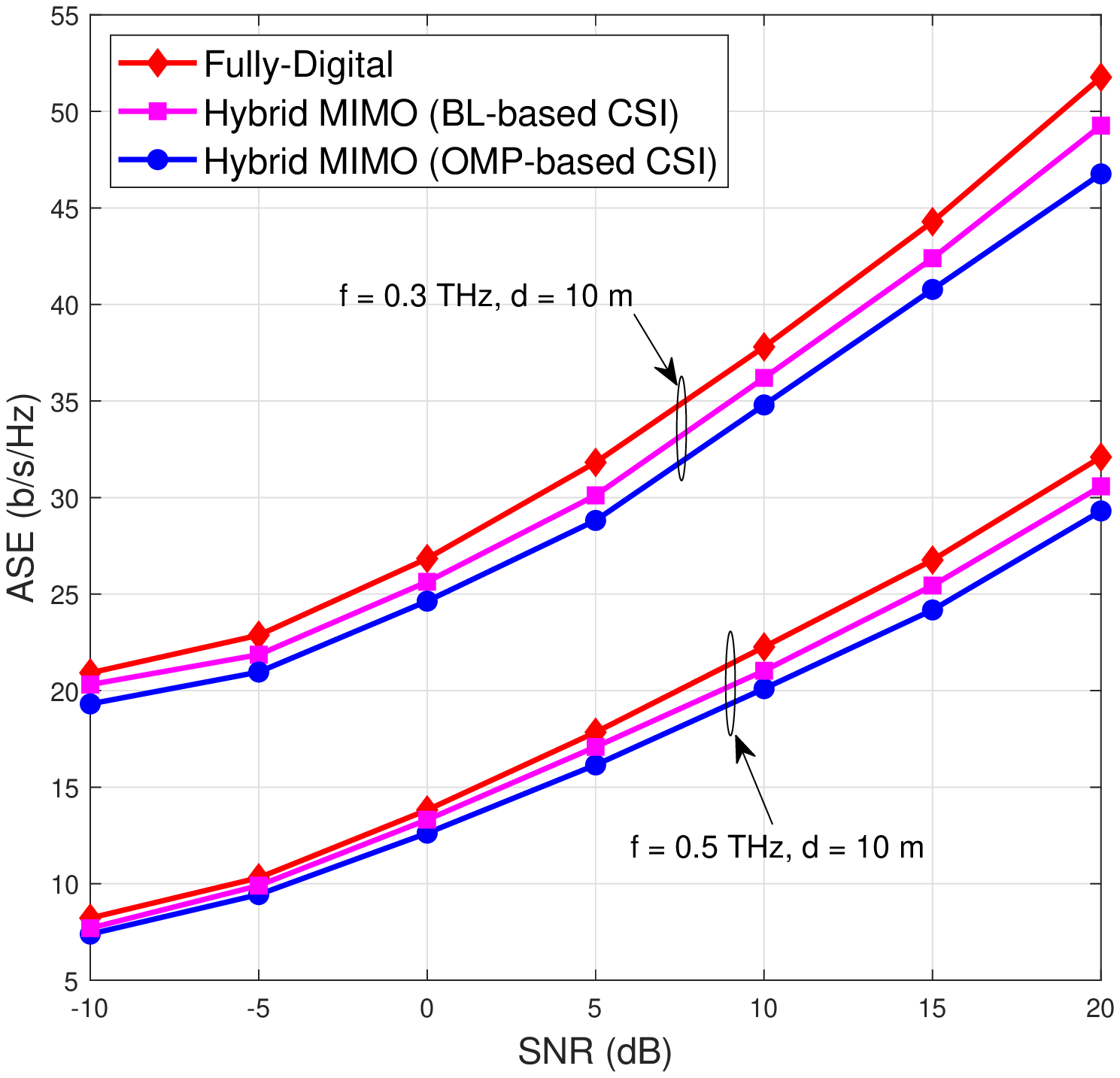}}
\hfil
\hspace{-10pt}\subfloat[]{\includegraphics[scale = .33]{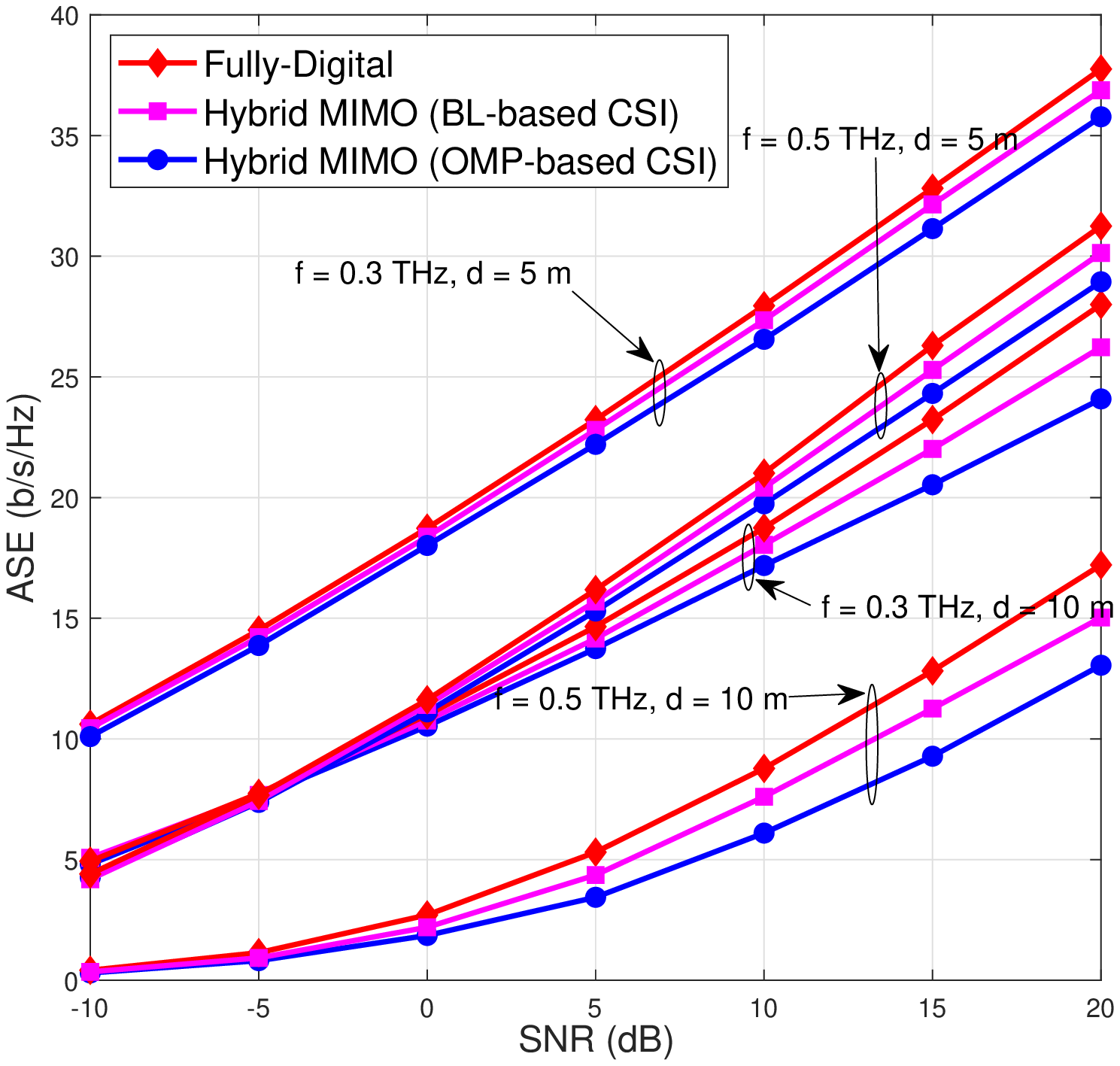}}
\hfil
\hspace{-10pt}\subfloat[]{\includegraphics[scale = .33]{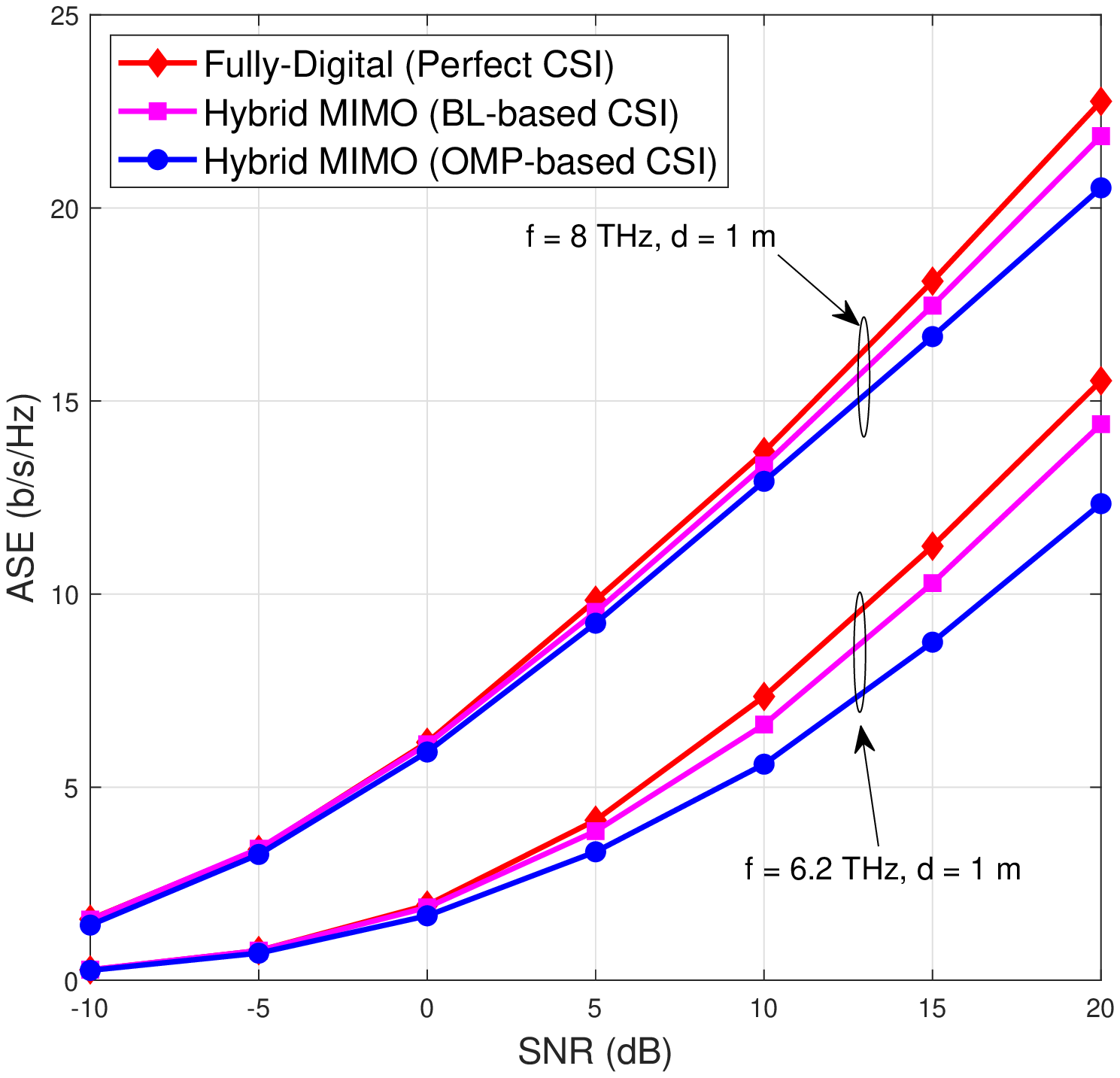}}
\caption{ASE versus SNR comparison for a \ac{THz} \ac{MIMO}, $ \left(a\right)$ System-I; $ \left(b\right)$ System-II, with different frequencies and distances; $ \left(c\right)$ Effect of molecular absorption losses on ASE \vspace{-10pt}}
\label{fig:THz_ch_capacity_plots}
\end{figure*}
\begin{figure*}[t!]
\centering
\subfloat[]{\includegraphics[scale = .45]{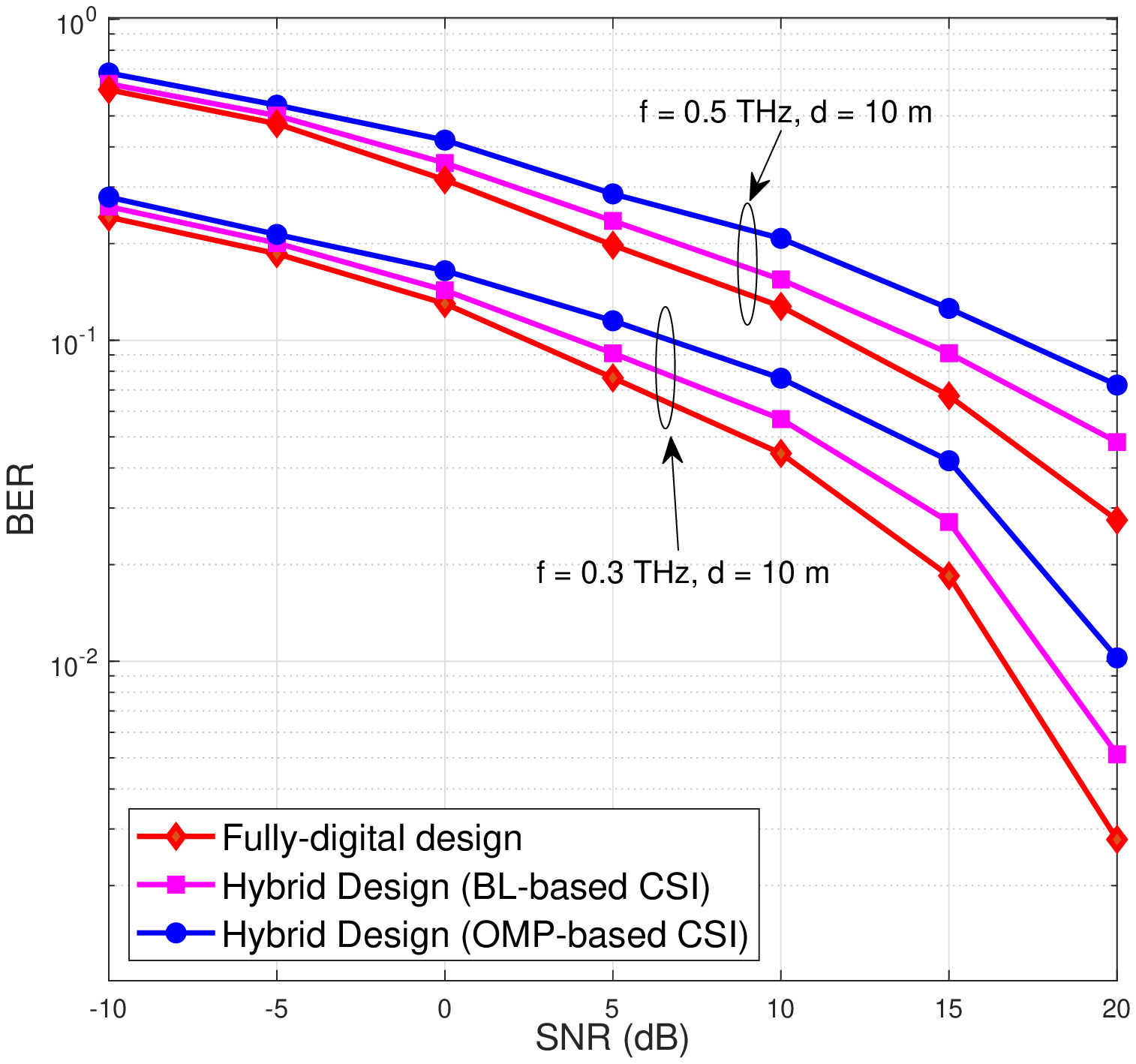}}
\hfil
\hspace{-20pt}\subfloat[]{\includegraphics[scale = .37]{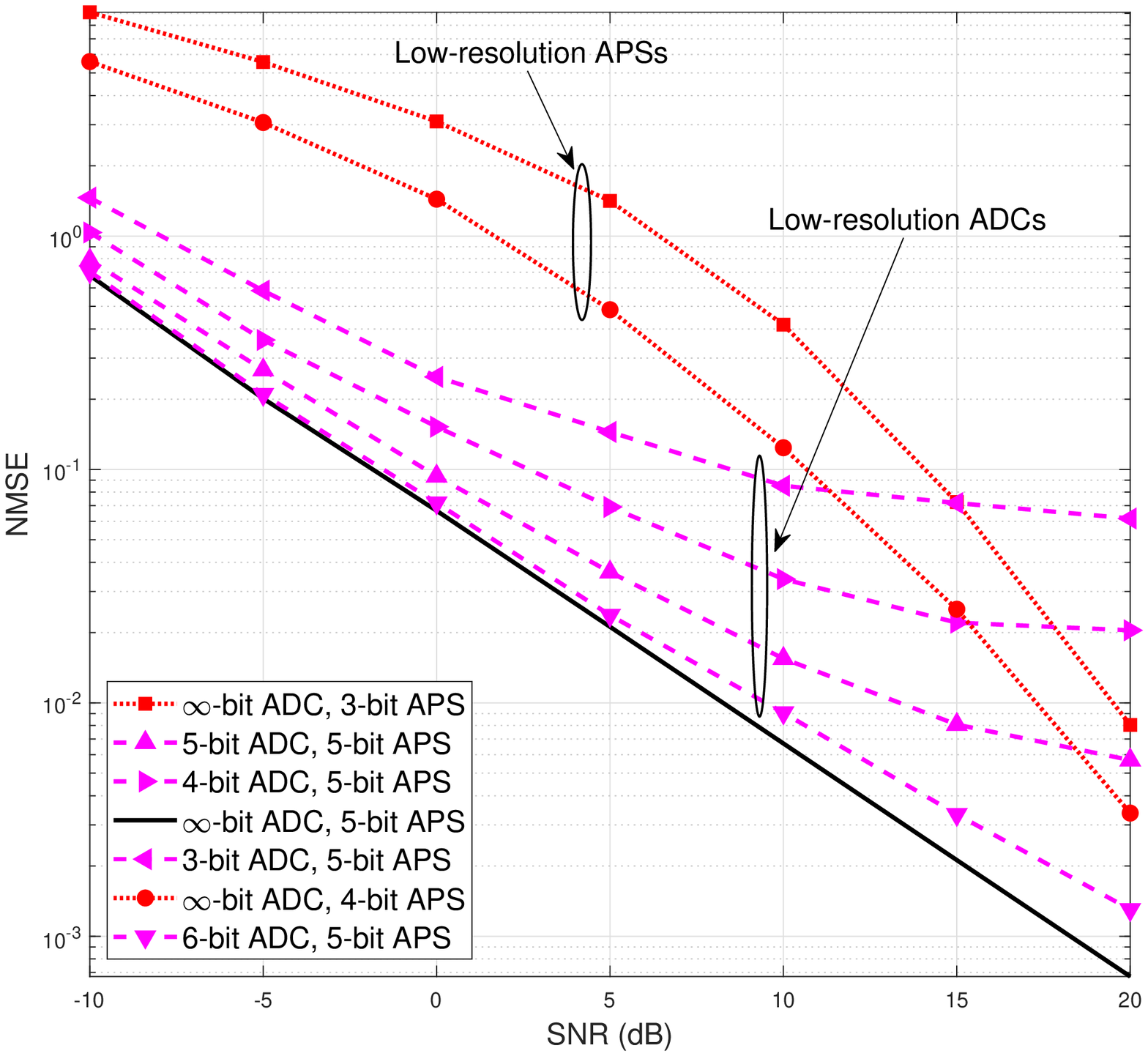}}
\caption{$ \left(a\right)$ \ac{BER} versus SNR comparison, for \ac{THz} \ac{MIMO}  System-II; $\left(b\right)$ NMSE versus SNR comparison for a \ac{THz} \ac{MIMO} System-I considering low-resolution ADCs and APSs. \vspace{-10pt}}
\label{fig:THz_ch_capacity_BER_plots}
\end{figure*}

\subsection{\ac{THz} \ac{MIMO} Hybrid Transceiver Design}
This subsection evaluates both the \ac{ASE} in bits/sec/Hz and the \ac{BER} to illustrate the performance of the proposed hybrid transceiver design. The \ac{ASE} is computed using the well-known Shannon capacity formula as $C=\log_2\left\vert \mathbf{I}_{N_S}+\frac{1}{N_S}\mathbf{R}_{n}^{-1}\mathbf{H}_{\text{eq}}\mathbf{H}_{\text{eq}}^H\right\vert,$
where the matrices $\mathbf{R}_{n}$ and $\mathbf{H}_{\text{eq}}$ denote the covariance of the combined noise and equivalent baseband channel, respectively, given by $\mathbf{R}_{n}=\sigma_v^2\bar{\mathbf{W}}_{\text{BB}}^H \bar{\mathbf{W}}_{\text{RF}}^H \bar{\mathbf{W}}_{\text{RF}}\bar{\mathbf{W}}_{\text{BB}},
\mathbf{H}_{\text{eq}}=\bar{\mathbf{W}}_{\text{BB}}^H \bar{\mathbf{W}}_{\text{RF}}^H \mathbf{H} \bar{\mathbf{F}}_{\text{RF}} \bar{\mathbf{F}}_{\text{BB}}.$ The quantities $\bar{\mathbf{W}}_{\text{RF}}, \bar{\mathbf{W}}_{\text{BB}}, \bar{\mathbf{F}}_{\text{RF}}$ and $\bar{\mathbf{F}}_{\text{BB}}$ have been evaluated using the proposed hybrid transceiver design described in Algorithm-\ref{algo3}, which in turn requires the estimated beamspace domain \ac{CSI} obtained from the \ac{OMP} (Algorithm-\ref{algo1}) or \ac{BL} (Algorithm-\ref{algo2}) schemes. The \ac{ASE} of a fully-digital \ac{THz} \ac{MIMO}  system having perfect CSI is also plotted therein to benchmark the performance and to demonstrate the gap between the proposed hybrid and ideal baseband transceiver architectures.  

Fig. \ref{fig:THz_ch_capacity_plots}(a) plots the \ac{ASE} versus \ac{SNR} for System-I. Observe that the proposed hybrid transceiver design using the estimated beamspace domain \ac{CSI} yields an \ac{ASE} that is reasonably close to that of the fully-digital system having perfect \ac{CSI}. This demonstrates the efficacy of the proposed hybrid transceiver design as well as that of the OMP and BL-based sparse CSI estimation techniques. The improved \ac{CSI} estimation accuracy of the \ac{BL} technique also leads to higher \ac{ASE} in comparison to the same achieved using \ac{OMP}-based CSI. Furthermore, the \ac{ASE} is also plotted for two different frequencies, viz., $f \in \{0.3,0.5\}$ THz. Observe from the figure that due to the high free-space losses characterized by \eqref{eq:losses}, the \ac{ASE} of the \ac{THz} \ac{MIMO} system at the higher operating frequency $f=0.5$ THz is lower than that at $f=0.3$ THz, for a given transmission distance of $d = 10$ m. A similar observation can be made in Fig. \ref{fig:THz_ch_capacity_plots}(b) for System-II, where the effect of varying the transmission distance is also presented. Once again, due to the high free-space losses, the \ac{ASE} of the \ac{THz} \ac{MIMO} system at the higher transmission distance of $d=10$ m is lower than at $d=5$ m. Fig. \ref{fig:THz_ch_capacity_plots}(c) illustrates another interesting result by considering the pair of frequencies $f \in \{6.2,8.0\}$ THz for the same transmission distance of $d = 1$ m. Note that the \ac{ASE} of the \ac{THz} \ac{MIMO} system for $f=6.2$ THz is lower than at $f=8.0$ THz, which is in contrast to the results of Figs. \ref{fig:THz_ch_capacity_plots}(a)-(b). Following the procedure described in Section-\ref{sec:k_abs} and employing the \ac{HITRAN} database, the molecular absorption coefficients $k_{\text{abs}}(f)$ at $f = 6.2$ THz and $8.0$ THz approximately evaluate to $\approx 3.1\text{ m}^{-1}$ and $\approx 0.3\text{ m}^{-1}$, respectively. Hence, the poor performance at $f=6.2$ THz can be attributed to the higher molecular absorption losses at this operating frequency, which is a characteristic feature of the \ac{THz} \ac{MIMO} channel. Therefore, in order to precisely characterize the system performance at a specific frequency, one must consider the effect of the molecular absorption coefficient $k_{\text{abs}}(f)$ and the associated losses $L_{\text{abs}}(f,d)$, as described in \eqref{eq:losses}. Finally, Fig. \ref{fig:THz_ch_capacity_BER_plots}(a) plots the \ac{BER} versus SNR  using the proposed hybrid transceiver design for quadrature-phase shift keying (QPSK) modulation. A similar trend is observed, where the proposed design using the BL-based estimated \ac{CSI} yields a BER sufficiently close to the benchmark. Furthermore, the BER of the \ac{THz} \ac{MIMO} system at $f=0.3$ THz is lower than at $f=0.5$ THz.

{
\subsection{Effects of Low-Resolution ADCs and APSs}
Fig. \ref{fig:THz_ch_capacity_BER_plots}(b) analyzes the effects of employing low-resolution ADCs on the CSI estimation performance of the proposed BL-based approaches. For this, the quantized pilot output $\mathbf{y}_{\mathrm{q}}$ corresponding to the pilot output $\mathbf{y}$ of \eqref{eq: CS_eq_model} is expressed as $\mathbf{y}_{\mathrm{q}} = \mathcal{Q}(\mathbf{y})$, where $\mathcal{Q}(\cdot)$ represents the  element-wise quantization operator. Hence, the $i$th element $\mathbf{y}_{\mathrm{q}}(i)$ of the quantized pilot output is given as $$\mathbf{y}_{\mathrm{q}}(i) = \mathcal{Q}( \mathrm{Real}[\mathbf{y}(i)]) + j\mathcal{Q}( \mathrm{Imag}[\mathbf{y}(i)]).$$ Note that for a $b_{\mathrm{q}}$-bit quantizer, the number of levels $N_{\mathrm{q}} = 2^{b_{\mathrm{q}}}$, which implies that the quantizer output $\mathcal{Q}(y)$ for any scalar $y \in \mathbb{R}$ is given as
\begin{align*} 
\mathcal{Q}(y) = \left\lbrace \begin{array}{ll} v_1, & {y\in [u_0, u_1];} \\ v_2, & {y\in (u_1, u_2];} \\ \vdots & \vdots \\ v_{N_{\mathrm{q}}}, & {y\in (u_{N_{\mathrm{q}}-1}, u_{N_{\mathrm{q}}}],} \end{array} \right.
\end{align*}
where $u_0 < u_1 < \cdots < u_{N_{\mathrm{q}}}$ denote the quantization thresholds, whereas $\left\{ v_i \right\}_{i = 1}^{N_{\mathrm{q}}}$ represent the quantizer output levels. For simplicity, we consider a uniform mid-point quantizer, which obeys
\begin{align*} 
&u_i = (-N_{\mathrm{q}}/2+i)\Delta, \,i= 0,\ldots, N_{\mathrm{q}},\\ 
&v_i = (u_{i-1} + u_i)/2, \,i=1, \ldots, N_{\mathrm{q}}, 
\end{align*}
where $\Delta$ denotes the quantization step-size. Furthermore, the model for the quantized pilot output $\mathbf{y}_{\mathrm{q}}$ can be expressed as $$\mathbf{y}_{\mathrm{q}} = \widetilde{\mathbf{\Phi}} \mathbf{h}_{b} + \mathbf{v} + \mathbf{v}_{\mathrm{q}},$$ where $\mathbf{v}_{\mathrm{q}}$ denotes the additional quantization noise. The NMSE performance of the proposed sparse channel estimation schemes considering different ADC resolutions is illustrated in Fig. \ref{fig:THz_ch_capacity_BER_plots}(b). One can readily observe that the NMSEs of the proposed techniques for $b_{\mathrm{q}} = 6$-bit ADC resolution are almost identical to that of the $\infty$-bit resolution, i.e. for the analog pilot outputs. Furthermore, the NMSE increases upon decreasing the ADC resolution, which is attributed to the increased quantization noise. However, for the low SNR regime of $-10$ dB to $10$ dB, which is a typical scenario in the THz band, the NMSEs achieved for $4$- and $3$-bit ADC resolutions are still acceptable. This demonstrates the feasibility of the proposed CSI estimation schemes for  practical THz hybrid MIMO systems also, which demand low-resolution ADCs due to their high bandwidth for the sake of reducing their power consumption.}

{Fig. \ref{fig:THz_ch_capacity_BER_plots}(b) also demonstrates the effects of using low-resolution APSs on the CSI estimation performance. Note that setting the RF TPC and RC using the DFT matrices requires $\log_2(N_T)$- and $\log_2(N_R)$-bit APSs, respectively. Thus, $5$-bit APSs are sufficient for an efficient sparse CSI estimation in a THz hybrid MIMO system having $N_T = N_R = 32$ antennas. Furthermore, the proposed CSI estimation model is general, and it can also operate with APSs having further low resolution of $3$-bit and $4$-bit, as seen in the Fig. \ref{fig:THz_ch_capacity_BER_plots}(b).}
\section{Conclusions} \label{sec:conclusion}
This work developed a practical \ac{MIMO} channel model considering several key aspects of the \ac{THz} band, such as the reflection losses and molecular absorption. Then a sparse \ac{CSI} estimation model was developed for exploiting the underlying angular-sparsity of the \ac{THz} \ac{MIMO} channel, followed by the \ac{OMP} and improved \ac{BL}-based frameworks for CSI estimation. Furthermore, the \ac{BCRLB} was also determined for benchmarking the performance of the proposed channel estimation techniques. Finally, optimal hybrid TPC and RC designs were developed, which directly employ the estimated beamspace domain \ac{CSI} and require only limited \ac{CSI} feedback. Our simulation setup employed practical \ac{THz} \ac{MIMO} channel parameters obtained from the \ac{HITRAN}-database. The proposed \ac{BL} framework was seen to yield both an \ac{MSE} performance close to the \ac{BCRLB} and an improved \ac{ASE}. Furthermore, the proposed frameworks require a reduced number of pilot beams for sparse signal recovery using compressed measurements. However, both the \ac{ASE} and \ac{BER} degraded upon increasing the frequency as well as the transmission distance, which became particularly pronounced at certain specific frequencies, where the molecular absorption was extremely high.
\appendices
\section{Proof of Lemma \ref{lem:lemma1}} \label{appen:proof_lem_1}
The total coherence $\mu^t\left(\widetilde{\mathbf \Phi}\right)$ of the equivalent sensing matrix $\widetilde{\mathbf \Phi}$ is defined as \cite{elad2007optimized,li2013projection} $$\mu^t\left(\widetilde{\mathbf \Phi}\right) = \sum_{i = 1}^{G_RG_T}\sum_{j=1, j \neq i}^{G_RG_T} \left\vert \widetilde{\mathbf \Phi}_{i}^H  \widetilde{\mathbf \Phi}_{j} \right\vert^2,$$ where the quantities $\widetilde{\mathbf \Phi}_{i}$ and $\widetilde{\mathbf \Phi}_{j}$ represent the $i$th and $j$th columns, respectively, of the matrix $\widetilde{\mathbf \Phi}$. Note that it can be bounded as follows: $$\mu^t\left(\widetilde{\mathbf \Phi}\right) \leq \left\Vert \widetilde{\mathbf \Phi} \widetilde{\mathbf \Phi}^H \right\Vert_F^2 = \sum_{i = 1}^{G_RG_T}\sum_{j=1}^{G_RG_T} \left\vert \widetilde{\mathbf \Phi}_{i}^H  \widetilde{\mathbf \Phi}_{j} \right\vert^2.$$  Substituting $\widetilde{\mathbf F} = \mathbf{X}_{p}^T \mathbf{F}_{\text{RF}}^T \mathbf A^{*}_T (\Phi_T)$ and $\widetilde{\mathbf W} = \mathbf W^H_{\text{BB}} \mathbf W^H_{\text{RF}} \mathbf A_R(\Phi_R)$ in \eqref{eq:eq_sensing_mat}, one can rewrite the above bound as
\begin{align}
\mu^t\left(\widetilde{\mathbf \Phi}\right) \leq \left\Vert \left(\widetilde{\mathbf F}\widetilde{\mathbf F}^H\right) \boldsymbol \otimes \left( \widetilde{\mathbf W}\widetilde{\mathbf W}^H\right) \right\Vert_F^2.
\end{align}
Furthermore, employing the relationship  $\left\Vert \mathbf{A} \boldsymbol \otimes \mathbf{B} \right\Vert_F^2 = \left\Vert \mathbf{A} \right\Vert_F^2 \left\Vert \mathbf{B} \right\Vert_F^2$, one can simplify the above expression as 
\begin{align}
\mu^t\left(\widetilde{\mathbf \Phi}\right) \leq  \left\Vert \widetilde{\mathbf F}\widetilde{\mathbf F}^H \right\Vert_F^2 \left\Vert  \widetilde{\mathbf W}\widetilde{\mathbf W}^H \right\Vert_F^2 = \frac{G_T}{N_T} \left\Vert \mathbf{X}_{p}^T \mathbf{X}_{p}^{*} \right\Vert_F^2\ \times \frac{G_R}{N_R}\left\Vert \mathbf{W}_{\text{BB}}^H \mathbf{W}_{\text{BB}} \right\Vert_F^2. \label{eq:total_coh_simplified_1}
\end{align}
The simplification in the above result exploits the semi-unitary property of the matrices $\mathbf A_T (\Phi_T)$ and $\mathbf A_R(\Phi_R)$, respectively, given in \eqref{eq:sem _unit}, and owing  to  the  choice of the RF TPC  $\mathbf{F}_{\text{RF}}$ and RC  $\mathbf{W}_{\text{RF}}$ as the DFT matrices. From \eqref{eq:total_coh_simplified_1}, it can be readily observed that minimization of the total coherence $\mu^t\left(\widetilde{\mathbf \Phi}\right)$ can be achieved by the minimization of the quantities $\left\Vert \mathbf{X}_{p}^T \mathbf{X}_{p}^{*} \right\Vert_F^2$ and $\left\Vert \mathbf{W}_{\text{BB}}^H \mathbf{W}_{\text{BB}} \right\Vert_F^2$, with respect to the pilot matrix $\mathbf{X}_{p}$ and the baseband RC matrix $\mathbf{W}_{\text{BB}}$, respectively. The optimal pilot matrix $\mathbf{X}_{p}$ subject to a suitable training power constraint can now be derived as follows. 

Note that minimization of $\left\Vert \mathbf{X}_{p}^T \mathbf{X}_{p}^{*} \right\Vert_F^2$ is equivalent to the  minimization of $\left\Vert \mathbf{X}_{p,i}^T \mathbf{X}_{p,i}^{*} \right\Vert_F^2$ with respect to each $\mathbf{X}_{p,i}, 1 \leq i \leq N_F,$ since the pilot matrix $\mathbf{X}_{p}$ is block diagonal. Therefore, the optimal pilot matrix design optimization problem can be formulated as
\begin{align}
\min_{\mathbf{X}_{p,i}} \left\Vert \mathbf{X}_{p,i}^T \mathbf{X}_{p,i}^{*} \right\Vert_F^2, \ \ 
\text{s.t.} \left\Vert \mathbf{X}_{p,i} \right\Vert_F^2 = \frac{M_T}{N_F}. \label{eq:pilot_design_opt}
\end{align}
The closed-form solution of the above problem can be derived as follows. Let $\mathbf{X}_{p,i} = \mathbf{U} \boldsymbol\Sigma \mathbf{V}_1^H$ represent the \ac{SVD} of the pilot matrix $\mathbf{X}_{p,i}$, where the matrices $\mathbf{U}$ and $\mathbf{V}_1$ are unitary matrices of size $N_\text{RF} \times N_\text{RF}$ and ${\frac{M_T}{N_F}} \times {\frac{M_T}{N_F}}$, respectively. Since $M_T \leq N_T$, which implies that $\frac{M_T}{N_F}=\frac{M_T N_\text{RF}}{N_T} \leq N_\text{RF}$, the singular matrix $\boldsymbol\Sigma \in \mathbb C^{N_\text{RF} \times \frac{M_T}{N_F}}$ has the following structure
\begin{align}
\boldsymbol\Sigma = \left[\mathrm{diag}\left(\sigma_1, \cdots, \sigma_{{\frac{M_T}{N_F}}}\right) \ \ \ \mathbf{0}_{{{\frac{M_T}{N_F}}} \times N_\text{RF}-{{\frac{M_T}{N_F}}}}\right]^T, \label{eq:singular_matrix_pilot}
\end{align} 
where $\sigma_1, \cdots, \sigma_{{\frac{M_T}{N_F}}}$ denote the singular values of the pilot matrix $\mathbf{X}_{p,i}$. Exploiting now the property of the unitary matrices $\mathbf{U}$ and $\mathbf{V}_1$, and the expression of the singular matrix $\boldsymbol\Sigma$ defined in \eqref{eq:singular_matrix_pilot}, the optimization problem in \eqref{eq:pilot_design_opt} can be reformulated as
\begin{align}
\min_{\sigma_i, 1 \leq \sigma_i \leq {\frac{M_T}{N_F}}} \sum_{i=1}^{{\frac{M_T}{N_F}}} \sigma_i^4, \ \ \  
\text{s.t.} \sum_{i=1}^{{\frac{M_T}{N_F}}} \sigma_i^2 = {\frac{M_T}{N_F}}.
\end{align}
Using the KKT conditions \cite{boyd2004convex} for solving the above optimization problem, the solution is obtained as $\sigma_i^{\text{opt}} =  1, \ 1\leq i \leq {\frac{M_T}{N_F}}$. Substituting the values of $\sigma_i^{\text{opt}}$ into the singular matrix of \eqref{eq:singular_matrix_pilot}, followed by employing the resultant expression in the \ac{SVD} of the pilot matrix $\mathbf{X}_{p,i}$ yields 
\begin{align}
\mathbf{X}_{p,i} = \mathbf{U} \Big[\mathbf{I}_{{\frac{M_T}{N_F}}} \ \ \mathbf{0}_{{{\frac{M_T}{N_F}}} \times N_\text{RF}-{{\frac{M_T}{N_F}}}}\Big]^T \mathbf{V}_1^H,
\end{align}
which is the desired result. The optimal RC matrix $\mathbf{W}_{\text{BB}}$ can also be derived following similar lines. $\blacksquare$
\section{Proof of Lemma \ref{lem:lemma2}} \label{appen:proof_lem_2}
The \ac{E-step} used for determining the conditional expectation of the log-likelihood function $\mathbf{\mathcal{L}}\left(\mathbf{\Gamma}| \widehat{\mathbf{\Gamma}}^{(j-1)}\right)$ is given as
\begin{align}
\mathbf{\mathcal{L}}\left(\mathbf{\Gamma}| \widehat{\mathbf{\Gamma}}^{(j-1)}\right)=&\mathbb{E}_{{{\mathbf{h}}_{b}\vert{\mathbf{y}}; \widehat{\mathbf{\Gamma}}^{(j-1)}}} \left\{\log f({ \mathbf{y}}, {\mathbf{h}}_{b}; {\mathbf{\Gamma}} )\right\}\nonumber\\
= &\mathbb{E}_{{{\mathbf{h}}_{b}|{\mathbf{y}}; \widehat{\mathbf{\Gamma}}^{(j-1)}}} \left\{ \log \left[f( {\mathbf{y}}| {\mathbf{h}}_{b})\right]\right\} +\mathbb{E}_{{{\mathbf{h}}_{b}|{\mathbf{y}}; \widehat{\mathbf{\Gamma}}^{(j-1)}}}\left\{\log \left[ f({ \mathbf{h}}_{b}; {\mathbf{\Gamma}})\right] \right\}.\label{eq: smv_gamma_update1}
\end{align}
The term inside the first $\mathbb{E}\{\cdot\}$ operator can be simplified as
\begin{align}
 \log f( {\mathbf{y}}| {\mathbf{h}}_{b}) = -M_TM_R\log(\pi) - \log\left[\det(\mathbf{R}_{v})\right] - \left({\mathbf{y} -  \widetilde{\mathbf \Phi}\mathbf h_{b}}\right)^H \mathbf{R}_{v}^{-1} \left({\mathbf{y}- \widetilde{\mathbf \Phi}\mathbf h_{b}}\right)\label{eq:first_term},
\end{align}
which is seen to be independent of the hyperparameter matrix $\boldsymbol \Gamma$. Therefore, the subsequent \ac{M-step} can ignore this term, while maximizing the likelihood $\mathbf{\mathcal{L}} \left( \mathbf{\Gamma}| \widehat{\mathbf{\Gamma}}^{(j-1)} \right)$ in \eqref{eq: smv_gamma_update1}. The equivalent optimization problem in the \ac{M-step} follows as
\begin{align}
\widehat{\boldsymbol\Gamma}^{(j)}=\arg \max_{\boldsymbol \Gamma} \mathbb{E}_{{{\mathbf{h}}_{b}|{\mathbf{y}}; \widehat{\mathbf{\Gamma}}^{(j-1)}}} &\left\{ \log \left[ f({ \mathbf{h}}_{b}; {\mathbf{\Gamma}})\right] \right\}. \label{eq: M-SIP_update1}
\end{align}
Upon substituting $f({ \mathbf{h}}_{b}; {\mathbf{\Gamma}})$ from \eqref{eq: gaussian_prior1} into the above optimization objective, 
the maximization problem can be decoupled into separate maximization problems with respect to the individual hyperparameters $\gamma_i$ as
\begin{align}
\widehat{\gamma}_i^{(j)} = \arg \max_{\gamma_i}  \left[- \log(\gamma_i)-\displaystyle\frac{\mathbb{E}_{{{\mathbf{h}}_{b}|{\mathbf{y}}; \widehat{\mathbf{\Gamma}}^{(j-1)}}} \left\{ \vert \mathbf h_{b}(i)\vert^2 \right\}}{\gamma_i} \right]\!.\!\! \label{eq: smv_gamma_update3_3}
\end{align}
Solving the above problem yields the estimates $\widehat{\gamma}_i^{(j)}$ as
\begin{align}
\widehat{\gamma}_i^{(j)}=\mathbb{E}_{{{\mathbf{h}}_{b}|{\mathbf{y}}; \widehat{\mathbf{\Gamma}}^{(j-1)}}} \left\{ \vert \mathbf h_{b}(i)\vert^2 \right\}. \label{eq: smv_gamma_update4}
\end{align}
To simplify the conditional expectation $\mathbb{E}_{{{\mathbf{h}}_{b}|{\mathbf{y}}; \widehat{\mathbf{\Gamma}}^{(j-1)}}} \left\{ \cdot \right\}$ above, the \textit{a posteriori} \ac{pdf} $f\left({{\mathbf{h}}_{b}|\mathbf{y}};\widehat{\boldsymbol\Gamma}^{(j-1)}\right)$ of $\mathbf{h}_{b}$ can be expressed as \cite{kay1993fundamentals}: $f\left({{\mathbf{h}}_{b}|\mathbf{y}};\widehat{\boldsymbol\Gamma}^{(j-1)}\right) = \mathcal{CN}\left( \boldsymbol{\mu}_{b}^{(j)}, \mathbf{R}_{b}^{(j)}\right),$
where the quantities $\boldsymbol{\mu}_{b}^{(j)} \in \mathbb{C}^{G_R G_T \times 1}$ and $\mathbf{R}_{b}^{(j)}  \in \mathbb{C}^{G_R G_T \times G_R G_T}$ are defined as
\begin{align}
    \boldsymbol{\mu}_{b}^{(j)}=\mathbf{R}_{b}^{(j)} \widetilde{\mathbf{\Phi}}^H \mathbf{R}_{v}^{-1} {\mathbf{y}},\ \ \ 
    \mathbf{R}_{b}^{(j)}= \left[ \widetilde{\mathbf{\Phi}}^H \mathbf{R}_{v}^{-1} \widetilde{\mathbf{\Phi}} + \left(\widehat{\boldsymbol\Gamma}^{(j-1)}\right)^{-1} \right]^{-1},
\end{align}
which represent the \textit{a posteriori} mean vector and covariance matrix, respectively, of the beamspace channel $\mathbf{h}_b$. Employing the \textit{a posteriori} \ac{pdf} $f\left({{\mathbf{h}}_{b}|\mathbf{y}};\widehat{\boldsymbol\Gamma}^{(j-1)}\right)$, the expression in \eqref{eq: smv_gamma_update4} can be simplified to  
\begin{align}
\widehat{\gamma}_i^{(j)} = \mathbf{R}_{b}^{(j)}(i,i) + \left\vert \boldsymbol{\mu}_{b}^{(j)}(i) \right\vert^2, \label{eq:hyperparameter_computation}
\end{align}
which is the desired expression. $\blacksquare$
\section{Derivation of the Hybrid \ac{MMSE} RC}
\label{appen:proof_MMSE_comb}
In order to derive the required expression, one can simplify the objective function of the optimization problem given in \eqref{eq:hybrid_combiner_optmzn} as 
\begin{align}
   \mathbb{E}\bigg\{ \left\Vert\bar{\mathbf x}-\bar{\mathbf{W}}_\text{BB}^H \bar{\mathbf{W}}_\text{RF}^H\mathbf y\right\Vert^{2}_{2}\bigg\} 
   &=\mathbb{E}\left\{ \mathrm{Tr}\left[\left( \bar{\mathbf x}-\bar{\mathbf{W}}_\text{BB}^H \bar{\mathbf{W}}_\text{RF}^H\mathbf y\right)\left( \bar{\mathbf x}-\bar{\mathbf{W}}_\text{BB}^H \bar{\mathbf{W}}_\text{RF}^H\mathbf y\right)^H\right] \right\}\nonumber\\
   &=\mathrm{Tr}\left[\mathbb{E}\left\{\bar{\mathbf x}\bar{\mathbf x}^H\right\}\right]-2\mathrm{Re}\left\{\mathrm{Tr}\left[\mathbb{E}\left\{\bar{\mathbf x}\mathbf{y}^H\right\}\bar{\mathbf{W}}_\text{RF}\bar{\mathbf{W}}_\text{BB}\right]\right\}\nonumber\\
   &\hspace{40pt}+\mathrm{Tr}\left[\bar{\mathbf{W}}_\text{BB}^H \bar{\mathbf{W}}_\text{RF}^H\mathbb{E}\left\{\mathbf{y}\mathbf{y}^H\right\}\bar{\mathbf{W}}_\text{RF}\bar{\mathbf{W}}_\text{BB}\right]. \label{eq:final_mmse_hyb_rc}
\end{align}
Since, the minimization is performed with respect to $\left(\bar{\mathbf{W}}_\text{RF},\bar{\mathbf{W}}_\text{BB}\right)$, one can neglect the first term $\mathrm{Tr}\left[\mathbb{E}\left\{\bar{\mathbf x}\bar{\mathbf x}^H\right\}\right]$. In order to further simplify \eqref{eq:final_mmse_hyb_rc}, one can add the constant term $\mathrm{Tr}\left[\bar{\mathbf{W}}_\text{M}^H \mathbb{E}\left\{\mathbf{y}\mathbf{y}^H\right\}\bar{\mathbf{W}}_\text{M}\right]$ in the above expression, where the optimal \ac{MMSE} RC $\bar{\mathbf{W}}_\text{M}$ obeys $\bar{\mathbf{W}}_\text{M}^H=\mathbb{E}\left\{\bar{\mathbf x}\mathbf{y}^H\right\}\mathbb{E}\left\{\mathbf{y}\mathbf{y}^H\right\}^{-1}$. Finally, the above expression can be reformulated as 
\begin{align}
   \mathbb{E}\bigg\{\left\Vert\bar{\mathbf x}-\bar{\mathbf{W}}_\text{BB}^H \bar{\mathbf{W}}_\text{RF}^H\mathbf y\right\Vert^{2}_{2}\bigg\}
   =&\mathrm{Tr}\left[\bar{\mathbf{W}}_\text{M}^H \mathbb{E}\left\{\mathbf{y}\mathbf{y}^H\right\}\bar{\mathbf{W}}_\text{M}\right]-2\mathrm{Re}\left\{\mathrm{Tr}\left[\bar{\mathbf{W}}_\text{M}^H\mathbb{E}\left\{\mathbf{y}\mathbf{y}^H\right\}\bar{\mathbf{W}}_\text{RF}\bar{\mathbf{W}}_\text{BB}\right]\right\}\nonumber\\
   & \hspace{50pt}+\mathrm{Tr}\left[\bar{\mathbf{W}}_\text{BB}^H \bar{\mathbf{W}}_\text{RF}^H\mathbb{E}\left\{\mathbf{y}\mathbf{y}^H\right\}\bar{\mathbf{W}}_\text{RF}\bar{\mathbf{W}}_\text{BB}\right]\nonumber\\
   =& \mathrm{Tr}\left[\left( \bar{\mathbf{W}}_\text{M}^H{-}\bar{\mathbf{W}}_\text{BB}^H \bar{\mathbf{W}}_\text{RF}^H\right)\mathbb{E}\left\{\mathbf{y}\mathbf{y}^H\right\}\left( \bar{\mathbf{W}}_\text{M}^H{-}\bar{\mathbf{W}}_\text{BB}^H \bar{\mathbf{W}}_\text{RF}^H\right)^H\right]\nonumber\\
   =&\left\Vert\mathbb{E}\left\{\mathbf{y}\mathbf{y}^H\right\}^{1/2}\left(\bar{\mathbf{W}}_\text{M}-\bar{\mathbf{W}}_\text{RF} \bar{\mathbf{W}}_\text{BB}\right)\right\Vert_{F}^2,
\end{align}
which is the desired expression. $\blacksquare$


\bibliographystyle{IEEEtran}
\bibliography{References}
\end{document}